\documentclass[a4paper,12pt]{article}

\usepackage[utf8]{inputenc} 
\usepackage[T1]{fontenc}    
\usepackage{lmodern}        
\usepackage[margin=1in]{geometry} 
\usepackage{microtype}      
\usepackage{amsmath, amssymb, amsfonts, amsthm, mathtools}
\usepackage{algorithm}
\usepackage{algpseudocode}

\usepackage{authblk}

\usepackage{graphicx}
\usepackage{hyperref}

\newtheorem{lemma}{\bf Lemma}[section]

\newtheorem{proposition}{\bf Proposition}[section]
\newtheorem{definition}{\bf Definition}[section]

\title{Wasserstein Geometry of Information Loss in Nonlinear Dynamical Systems}

\author[1,3]{Yiting Duan}
\author[2]{Zhikun Zhang}
\author[3]{Yi Guo\thanks{Corresponding author: \texttt{yg@westernsydney.edu.au}}}

\affil[1]{Translational Health Research Institute, Western Sydney University, Australia}
\affil[2]{School of Mathematics and Statistics, Northwestern Polytechnical University, Xi'an, China}
\affil[3]{Centre for Research in Mathematics and Data Science, Western Sydney University, Australia}
\date{\vspace{-5ex}}
\begin{document}

\maketitle

\begin{abstract}
Time-delay embedding is a powerful technique for reconstructing the dynamics of nonlinear systems. However, the reconstruction map is not always an embedding, a condition rarely verified in practice. When the reconstruction map is non-injective, multiple latent states may map to the same reconstructed state, leading to multi-valued $n$-step evolution. Consequently, the induced system no longer admits a deterministic closure, and the dispersion of future trajectories leads to ambiguity. In this work, we establish a measure-theoretic framework to quantify the ambiguity induced by multi-valued evolution and introduce intrinsic stochasticity to quantify the ambiguity over a finite horizon. For numerical implementation, we use the $k$-nearest-neighbor estimator to approximate intrinsic stochasticity under finite-resolution and finite-sampling settings. Numerical experiments on the synthetic and real-world datasets are consistent with the expectation: reconstructions closer to deterministic closure tend to produce lower scores, and deterministic predictors that take reconstructions with lower empirical closure scores as input are associated with lower rollout errors, suggesting that intrinsic stochasticity provides a new perspective for understanding failures of reconstruction and serves as a diagnostic for selecting reconstruction maps. 
\end{abstract}

\vspace{1em}

\section{Introduction}
Reconstructing nonlinear dynamics from partially observed trajectories via time-delay embedding has a rich history for signal processing and data-driven modelling \cite{roux1983observation, broomhead1986extracting, hamilton2020time, kamb2020time, brunton2016discovering, brunton2017chaos, kutz2022parsimony, champion2019data, hirsh2021structured} and has been widely used in nonlinear time-series tasks, including forecasting, causal detection, and equation discovery \cite{Stark1999, Stark2003, bakarji2023discovering, sugihara2012detecting, ye2015distinguishing, gao2023causal}. Most such pipelines implicitly rely on the assumption that the reconstruction map is injective and can separate relevant states on the original attractor. However, in practice, this assumption is rarely checked. Under certain conditions \cite{takens2006detecting, sauer1991embedology}, Takens' theorem justifies that a generic observation function yields a smooth embedding, i.e., a continuous injective map from the underlying state space to the reconstructed space. However, this guarantee is generic rather than a certificate for a given dataset. With fixed measurements and a finite sampling frequency, the resulting time-delay map may fail to be injective on the relevant part of the attractor. When the time-delay map is non-injective, distinct latent states may share the same reconstructed state. If these latent states lead to distinct future states, then the reconstructed dynamics no longer admit a single-valued deterministic closure. Downstream models, e.g., the predictive models detailed in Section~\ref{Downstream task impact}, that assume such deterministic closure can then suffer an irreducible ambiguity, since the same reconstructed state may correspond to multiple incompatible future states.   

A typical example of the downstream impact of unequal reconstruction quality is Convergent Cross Mapping (CCM), a widely used causal discovery method for nonlinear time series~\cite {sugihara2012detecting}. The core idea is that if a state variable $X$ drives $Y$, then $Y$ encodes information about $X$. Thus, the time-delay reconstruction of $Y (\mathcal{M}_Y)$ can be used to make accurate predictions of the state $\mathcal{M}_X$ via a $k$-nearest neighbor ($k$-NN) approach. As the data length increases, the prediction skill (correlation $\rho$) is expected to approach a high plateau near 1 gradually. However, as demonstrated in Figure~\ref{Misleading CCM result}, if one of the time-delay maps is poorly conditioned, rendering the time-delay maps unequally informative for the two observables, CCM can yield asymmetric cross-map skill and thereby mislead causal interpretation.

\begin{figure}[!ht] 
    \centering  
    \includegraphics[width=1\linewidth]{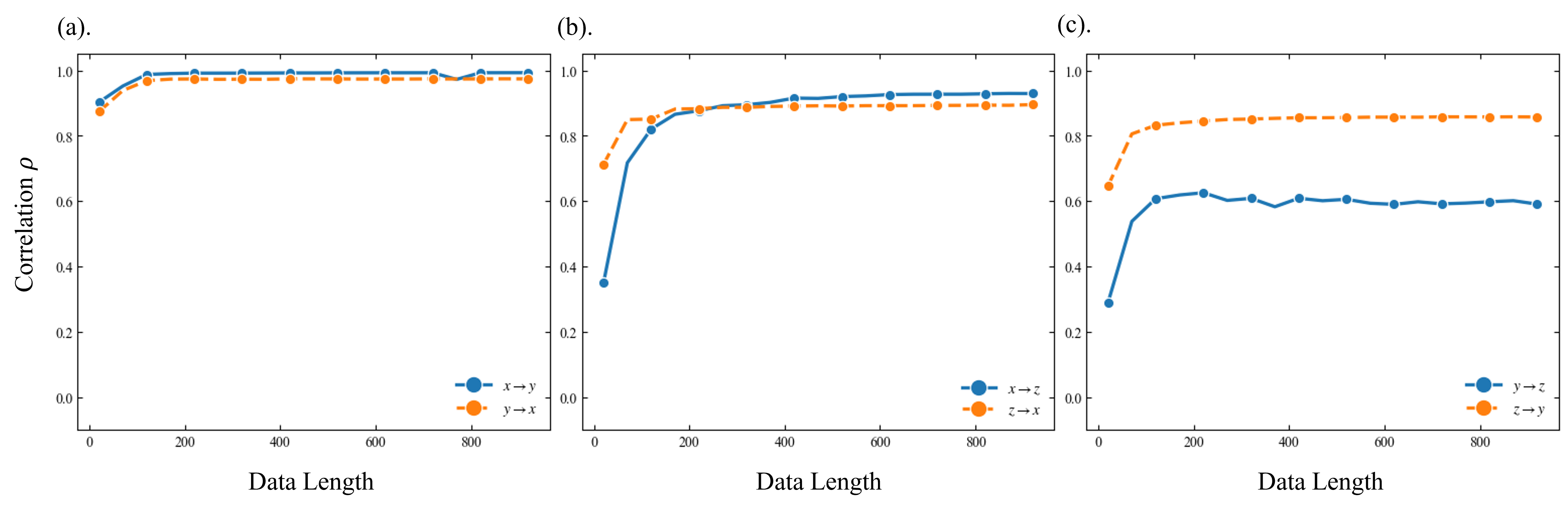}
    \caption{Performance of CCM on the R\"ossler system with dynamically coupled state variables, where CCM results should yield high cross-map scores in both directions when the time-delay reconstructions are sufficiently informative. The $Z$-based reconstruction is less informative in this example, which is consistent with the experiment results in Table~\ref{tab:future_cloud_baseline_comparison}. Each panel reports the cross-map correlation $\rho$ as a function of data length. (a). $X \Leftrightarrow Y$: both directions rapidly converge to a plateau near 1. (b). $X \Leftrightarrow Z$: both directions converge more slowly but still reach a high plateau. (c). $Y \Leftrightarrow Z$: cross-map skill is asymmetric $Z \rightarrow Y$ high, $Y \rightarrow Z$ low, indicating that the chosen time-delay reconstruction is not equally informative for both variables and can lead to information loss.} 
    \label{Misleading CCM result}
\end{figure}

Considerable effort has been devoted to analyzing the time-delay reconstructions, including operator-theoretic analyses of spectral properties~\cite{giannakis2021delay, korda2018linear, das2019delay} and practical criteria for selecting embedding parameters~\cite{pecora1995statistics, fraser1986independent}. These works provide important tools for understanding time-delay reconstructions, but they do not consider the case in which a fixed observation function yields a non-injective reconstruction. In the coordinate-projection setting considered here, different observed variables can produce reconstructions with different degrees of deterministic closure, and the difference in deterministic closure between observables is also closely correlated with topics such as the controllability or observability of complex networks~\cite{liu2013observability, liu2011controllability}.

Several efforts have also been made to investigate the quality of the time-delay reconstruction. When the system equation is available, the non-injectivity of the time-delay map is closely related to the nonlinear observability of the observed variable, which can be framed in differential geometry via the rank conditions on Lie derivatives~\cite{1101601, aguirre2005observability}. However, this framework is not data-driven, since it requires either explicit governing equations or stable Jacobian/derivative estimation. While other methods, such as SVDO~\cite {aguirre2011investigating} and DDA~\cite {gonzalez2020assessing}, are heuristic diagnostics that do not directly measure the information loss caused by non-injectivity, it remains a challenge to propose a data-driven diagnostic for measuring the ambiguity introduced by the non-deterministic nature of time-delay reconstructions.

\begin{figure}[!ht]
    \centering
    \includegraphics[width=1\linewidth]{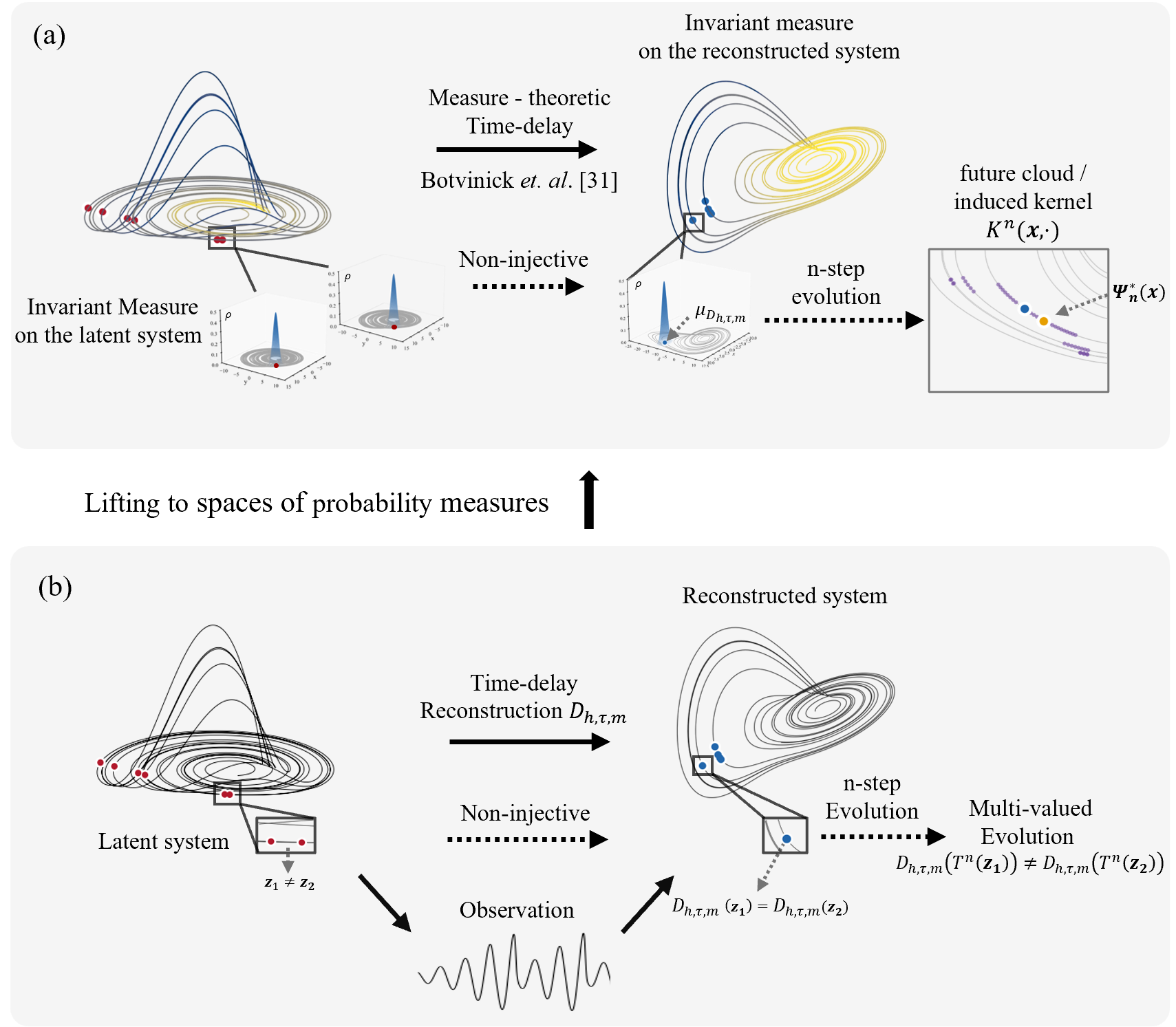}
    \caption{Overview of our measure-theoretic pipeline. (b) In the original state space, the time-delay map can be non-injective: distinct latent states $\mathbf{z}_1\neq \mathbf{z}_2$ may share the same reconstructed state $D_{h,\tau,m}(\mathbf{z}_1) = D_{h,\tau,m}(\mathbf{z}_2)$. If their $n$-step pushforward futures are different, e.g., $D_{h,\tau,m}(T^n(\mathbf{z}_{1})) \neq D_{h,\tau,m}(T^n(\mathbf{z}_{2}))$, then the induced flow does not hold a single-valued deterministic closure on the reconstructed space. (a) After lifting the dynamics to the space of probability measures, the $n$-step evolution of $\mathbf{x}$ is described by the induced kernel $K^n(\mathbf{x},\cdot)$ (the purple future cloud). When $\mathbf{x}=D_{h,\tau,m}(\mathbf{z}_1)=D_{h,\tau,m}(\mathbf{z}_2)$, this kernel can be non-degenerate or multi-modal. The intrinsic stochasticity $\mathcal{E}^*_n$ measures the deviation of $K^n(\mathbf{x},\cdot)$ from a Dirac mass, estimated by measuring the dispersion around the optimal Fr\'echet median $\Psi_n^*(\mathbf{x})$ (the orange point).}
    \label{fig:summary}
\end{figure}

Our motivation stems from the fact that dissipative systems typically admit an invariant measure, e.g., the Sinai-Ruelle-Bowen (SRB) measure, to which long-time empirical statistics along typical trajectories converge~\cite{sauer1991embedology, eckmann1985ergodic, botvinick2025invariant, botvinick2025measure}. Thus, the invariant measure on the attractor manifold provides a natural measurement for quantifying the ambiguity caused by the non-deterministic evolution of the reconstructed system. When the induced flow fails to be deterministic after an $n$-step evolution, the evolution can be represented by an induced kernel, with respect to the invariant measure, rather than as a single-valued map. In contrast, the kernel degenerates to a Dirac mass when the reconstructed dynamics admit deterministic closure after an $n$-step evolution. Moreover, another benefit is that, within the invariant measure framework, the question is well-defined, even when the geometry of the attractor manifold is highly complicated, e.g., fractal sets of zero Lebesgue measure.

In this work, we connect dynamical system reconstruction to optimal transport by lifting the induced dynamics to probability space, as illustrated in Figure~\ref{fig:summary}. We investigate the geometric mechanism of information loss: the information loss caused by the non-deterministic closure admits a mass-separation description, which is governed by the local competition between the dynamical stretching and observation curvature. To quantify the ambiguity caused by the $n$-step non-deterministic evolution, we introduce intrinsic stochasticity $\mathcal{E}_{n}^{*}$, a measure-theoretic, optimal-transport-based scalar functional defined as the average of the minimal Wasserstein-1 discrepancy between the induced $n$-step future kernel $K^{n}(\mathbf{x},\cdot)$ and its best deterministic approximation. Finally, we propose a $k$-NN based estimator $\widehat{\mathcal{E}}^{*}_{n,k,N}$ to approximate $\mathcal{E}_{n}^{*}$ under finite-resolution and finite-sample settings and numerical experiments illustrate its use for comparing candidates with different closure properties and provide a perspective to understand the failure of reconstruction maps. 

\section{Problem formulation and measure-theoretic lifting}

This section introduces the dynamical and measure-theoretic setting used throughout the paper. Lifting dynamical systems to measure spaces is a typical technique in data-driven modeling; for example, the Perron-Frobenius operator provides a measure-theoretic linear operator of the dynamical system, which propagates the evolution of the probability distributions in time \cite{brunton2021modern,yang2023optimal,botvinick2025measure}. We study the deterministic evolution of the dynamics induced by a fixed reconstruction map $F:\mathcal{M}\rightarrow \mathcal{X}$, and here we provide a brief introduction and establish standard notation on which we build in Section~\ref{Section three}. 

\subsection{Dynamical system reconstruction}\label{sec: Time-delay embedding}

Suppose the following dynamical system:
\begin{equation}\label{eq:flow_ode}
\frac{\rm{d}\mathbf{z}}{\rm{d}t} = f(\mathbf{z}),
\end{equation}
where $\mathbf{z}\in\mathcal{M}$, $\mathcal{M}$ is a $d$-dimensional compact $C^r$ manifold, and we denote the generated flow as $\{\Phi^t\}_{t\in\mathbb{R}}$ which admits an ergodic physical probability measure $\nu$ that captures the asymptotic statistical behavior of long-time trajectories. Typically, observations are measured at a fixed interval $\Delta t>0$ by a measurement function $h(\mathbf{z})\in\mathbb{R}^{l}$, and the induced discrete-time dynamical system is denoted as
\begin{equation}\label{eq:sampling_map}
T := \Phi^{\Delta t}.
\end{equation}
In general, the choice of the measurement function is limited by practical constraints, and the dimensionality may also be smaller than $d$. In this work, we are primarily interested in the case of $l = 1$, and the measurement function is the coordinate projection of the dynamical variable (e.g., $h(z_1, z_2, z_3) = z_1$); in other words, we obtain a single one-dimensional time-series $\{y_k\}_{k\ge 0}$. 

Throughout the paper, let $F: \mathcal{M} \to \mathcal{X}\subseteq\mathbb{R}^{m}$ denote a reconstruction map induced by observations. The time-delay map $D_{h, \tau, m}: \mathcal{M} \to \mathcal{X}\subseteq\mathbb{R}^{m}$ is defined by:  
\begin{equation}\label{eq:delay_map_def}
D_{h, \tau, m}(\mathbf{z})=\big(h(\mathbf{z}),\,h(T^{-\tau}\mathbf{z}),\,\ldots,\,h(T^{-(m-1)\tau}\mathbf{z})\big)^\top,
\end{equation}
and is one of the most popular choices for system reconstruction. Here, $\mathcal{X}$ denotes the reconstructed space equipped with the Euclidean metric $d_\mathcal{X}(\mathbf{x},\mathbf{x}')=\|\mathbf{x}-\mathbf{x}'\|_2$. This construction has several benefits, and most notably, Takens' embedding theorem guarantees that for a generic observable $h$ and a sufficiently large embedding dimension, the time-delay map is an embedding. 

However, this theorem does not guarantee that the time-delay map constructed from a fixed measurement function is an embedding. The genericity assumption in Takens' theorem means that the set of embedding observables is large in the chosen function space, typically in the Baire category sense, but it does not provide a certificate for any particular observable. An arbitrarily small smooth perturbation may avoid an inadequate candidate, which makes sense as a theoretical technique but is limited in practice when a fixed, given measurement function generates the dataset. For the embedding dimension, the theorem requires the dimension of the reconstructed delay coordinate, rather than implying that any arbitrary stack of $m$ consecutive samples achieves an embedding. For autonomous systems, successive time delays may be strongly correlated, especially when the lag value is small. On the other hand, when the time delay is too large, the reconstructed coordinates become decorrelated, and the time delay forms a degenerate, product-like geometry. An example is shown in Figure~\ref{fig:Increasing check}. However, the selection of delay parameters is nontrivial: increasing $m$ does not guarantee a successful embedding, while the choice of lag $\tau$ is also crucial. When $m$ is insufficient, or $h$ is degenerate, $D_{h, \tau, m}$ may be non-injective, mapping distinct states $\mathbf{z}_1\neq \mathbf{z}_2$ to the same state as $x = D_{h, \tau, m}(\mathbf{z}_1)=D_{h, \tau, m}(\mathbf{z}_2)$. 
\begin{figure}[!ht]
    \centering
    \includegraphics[width=1\linewidth]{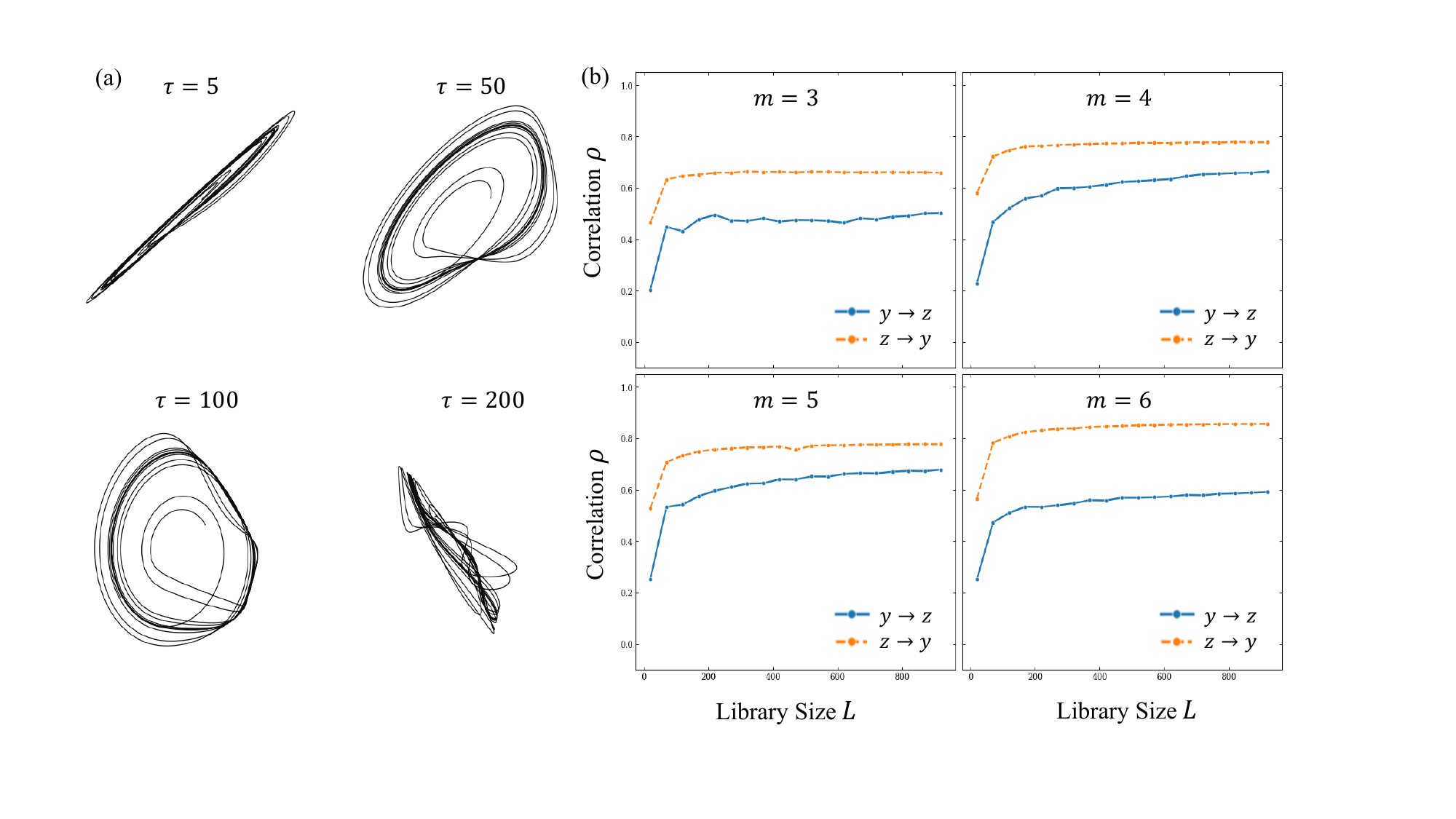}
    \caption{(a) The stretch-and-fold phenomenon as the lag value $\tau$ is increased for the R\"ossler system using the time-delay map induced by the $z_1$-coordinate projection. (b). Using the same CCM example as in Figure~\ref{Misleading CCM result}, increasing the embedding dimension for time-delay reconstructions induced by both $z_1$ and $z_3$ from $3$ to $6$ does not eliminate the non-informative phenomenon in the cross-map skill for $Y \Leftrightarrow Z$. In this example, increasing $m$ does not eliminate the asymmetric cross-map skill, indicating that larger embedding dimension alone does not necessarily resolve the finite-resolution closure problem.}
    \label{fig:Increasing check}
\end{figure}

When the governing equations are available, we also consider the differential map $D_{h,m}: \mathcal{M} \rightarrow \mathcal{X} \subseteq \mathbb{R}^m$ defined by:
\begin{equation}\label{derivative coordinate map}
    \begin{aligned}
    D_{h,m}(\mathbf{z})=\big(h(\mathbf{z}),\, \mathcal{L}_{f} h(\mathbf{z}),\, \ldots,\, \mathcal{L}_{f}^{m-1} h(\mathbf{z})\big),
    \end{aligned}
\end{equation} 
where $\mathcal{L}_{f} h=\nabla h\cdot f$ is the Lie derivative of $h$ along the vector field $f$. 

\subsection{Measure-theoretic lifting and induced kernels}

Suppose $\nu$ is an invariant measure of the dynamical system Eq.~\eqref{eq:flow_ode} on $\mathcal{M}$, and let $\mu_F:= F\#\nu \in \mathcal{P}(\mathcal{X})$ be the corresponding pushforward measure under $F$. For each reconstructed state $\mathbf{x}$, the corresponding fibre is:
\begin{equation}\label{eq:fibre}
F^{-1}(\mathbf{x})=\{\mathbf{z}\in\mathcal{M}\,:\,F(\mathbf{z})=\mathbf{x}\},
\end{equation}   
which may contain multiple latent states when $F$ is non-injective, such that the same $\mathbf{x}$ may correspond to finite distinct latent states. By the disintegration theorem~\cite{ambrosio2005gradient}, there exists a $\mu_F$-a.e. uniquely determined family of conditional probability measures $\{\nu_{\mathbf{x},F}\}_{\mathbf{x}\in \mathcal{X}} \subset \mathcal{P}(\mathcal{M})$ supported on the fibres $F^{-1}(\mathbf{x})$, such that, for any bounded Borel measurable map $\varphi:\mathcal{M}\to\mathbb{R}$, we have
\begin{equation}\label{eq:disintegration}
\int_{\mathcal{M}}\varphi(\mathbf{z})\,\rm{d}\nu(\mathbf{z})
=
\int_\mathcal{X}\left(\int_{F^{-1}(\mathbf{x})}\varphi(\mathbf{z})\,\rm{d}\nu_{\mathbf{x},F}(\mathbf{z})\right)\rm{d}\mu_F(\mathbf{x}).
\end{equation}
For any horizon $n\ge 1$, the induced $n$-step future kernel $K_F^n(\mathbf{x}, \cdot)$ given the current state $\mathbf{x}$ is defined as:
\begin{equation}\label{eq:kernel_def}
K_F^n(\mathbf{x},\cdot) := (F\circ T^n)\#\nu^F_\mathbf{x}.
\end{equation}
Equivalently, for any bounded Borel measurable map $\psi:\mathcal{X}\to\mathbb{R}$, we have
\begin{equation}\label{eq:kernel_action}
\int_\mathcal{X} \psi(\mathbf{x}')\,K_F^n(\mathbf{x},\rm{d}\mathbf{x}')
=
\int_{F^{-1}(\mathbf{x})} \psi(F(T^n \mathbf{z}))\,\rm{d}\nu_{\mathbf{x},F}(\mathbf{z}).
\end{equation}

If the time-delay map $F$ is injective on a set of full $\nu$-measure, then $\nu_{\mathbf{x},F}$ is a Dirac mass for $\mu_F$-a.e. $\mathbf{x}$, and consequently $K_F^n(\mathbf{x},\cdot)$ is concentrated at the unique $n$-step future. Conversely, if the conditional measure $\nu_{\mathbf{x},F}$ is non-Dirac and \(F\circ T^n\) is not $\nu_{\mathbf{x},F}$-a.e. constant on a set of positive $\mu_F$-measure, then $K_F^{n}(\mathbf{x},\cdot)$ is non-degenerate. The spread of the kernel $K_F^{n}(\mathbf{x},\cdot)$ reflects the ambiguity in the $n$-step future arising from the failure of a single-valued, deterministic closure.  

\section{Intrinsic Stochasticity}\label{Section three}



\subsection{Intrinsic Stochasticity}
\begin{definition}[Intrinsic Stochasticity]\label{def:En_pointwise}

Assume $\mathcal{X}$ is a Polish metric space and $K_F^n(\mathbf{x}, \cdot)\in \mathcal{P}_1(\mathcal{X})$ for $\mu_F$-a.e. $\mathbf{x}$. For each $n\ge 1$ and $y\in\mathcal{X}$, define the pointwise $\mathcal{W}_1$ approximation loss:
\begin{equation}
\ell_n(\mathbf{x},\mathbf{y}) := \int_{\mathcal X} d_{\mathcal X}(y,y')K^n_F(x,dy')
=
W_1\!\left(K^n_F(x,\cdot),\delta_y\right).
\end{equation}
For each $n\ge1$, we define the minimal pointwise risk:
\begin{equation}\label{eq:intrinsic_stochasticity}
m_{n,F}(\mathbf{x}):=\inf_{\mathbf{y}\in \mathcal{X}}\ell_n(\mathbf{x},\mathbf{y})
\end{equation}
and the Intrinsic Stochasticity is defined as: 
\begin{equation}\label{eq:En_def}
\mathcal E_{n,F}^{*} := \int_{\mathcal{X}} m_{n,F}(\mathbf{x})\,\rm{d}\mu_F(\mathbf{x}). 
\end{equation}
When the infimum is attained, the minimizer is the Fr\'echet/Geometric Median $\Psi_n^*(\mathbf{x})$: 
\begin{equation}
    \Psi_n^*(\mathbf{x})\in\arg\min_{\mathbf{y}\in \mathcal{X}}\ell_n(\mathbf{x},\mathbf{y}).
\end{equation}
\end{definition}

$\mathcal{E}^{*}_{n,F}$ quantifies the information loss arising from the non-deterministic closure. Since $\ell_n(\mathbf{x}, \mathbf{y})$ is based on $d_\mathcal{X}$, the corresponding $\mathcal{W}_1$-optimal deterministic prediction is the geometric median $\Psi_n^*(\mathbf{x})$. When the induced flow is non-deterministic, the kernel $K_{F}^n(\mathbf{x},\cdot)$ may become multimodal, e.g., splitting into two distinct branches. Here, branches refer to the local inverse branches of $F$, and their corresponding $n$-step pushforward futures under $F\circ T^n$. The conditional mean of $K^n_F(\mathbf x,\cdot)$, which minimizes the squared prediction risk, may lie between branches, possibly outside the reconstructed attractor or in a region with negligible invariant density. In contrast, a geometric median $\Psi_n^*$ is a natural $\mathcal{W}_1$-optimal point representative and is less sensitive to remote branches. For clarity, when the reconstruction map is fixed, we suppress the subscript $F$ and write $K^n(\mathbf{x},\cdot)$, $m_{n}$, $\mu$, $\mathcal{E}_{n}^*$.

\begin{proposition}\label{prop:intrinsic_stochasticity}
\[
\mathcal E_n^*=0
\quad\Longleftrightarrow\quad
K^n(\mathbf x,\cdot)\text{ is a Dirac measure for }\mu\text{-a.e. }\mathbf x .
\]
\end{proposition}

\begin{proof}
By Definition~\eqref{def:En_pointwise} and Eq.~\eqref{eq:En_def}, $\mathcal E_n^*=0$ implies $m_n(\mathbf x)=0$ for $\mu$-a.e. $\mathbf x$. Fix an $\mathbf x$, and let $P:=K^n(\mathbf x,\cdot)$, then for each $\varepsilon>0$, there exists a $\mathbf y_\varepsilon\in\mathcal X$ such that
\begin{equation}
\int_{\mathcal X}
d_{\mathcal X}(\mathbf y_\varepsilon,\mathbf u)\,
P(\mathrm d\mathbf u)
<\varepsilon.    
\end{equation}
For any $\mathbf{u}, \mathbf{v} \in \mathcal{X}$, we have $d_{\mathcal X}(\mathbf u,\mathbf v)\le d_{\mathcal X}(\mathbf u,\mathbf y_\varepsilon)+d_{\mathcal X}(\mathbf y_\varepsilon,\mathbf v)$ by triangle inequality and integrating over $\mathcal X\times\mathcal X$ with respect to \(P(\mathrm d\mathbf u)P(\mathrm d\mathbf v)\) yields
\begin{equation}
\int_{\mathcal X}\int_{\mathcal X} d_{\mathcal X}(\mathbf u,\mathbf v)\, P(\mathrm d\mathbf u)P(\mathrm d\mathbf v) \le 2\int_{\mathcal X}d_{\mathcal X}(\mathbf y_\varepsilon,\mathbf u)\,P(\mathrm d\mathbf u)<2\varepsilon.
\end{equation}
Letting \(\varepsilon\to0\), we obtain $\int_{\mathcal X}\int_{\mathcal X}d_{\mathcal X}(\mathbf u,\mathbf v)\, P(\mathrm d\mathbf u)P(\mathrm d\mathbf v)=0$ and \(d_{\mathcal X}(\mathbf u,\mathbf v)=0\) for
\(P\otimes P\)-a.e. \((\mathbf u,\mathbf v)\). Since \(d_{\mathcal X}\) is a
metric, \(P\otimes P\) is concentrated on the diagonal, which implies that
\(P\) is a Dirac measure. Therefore
\(K^n(\mathbf x,\cdot)\) is Dirac for \(\mu\)-a.e. \(\mathbf x\).

Conversely, suppose that $K^n(\mathbf x,\cdot)$ is a Dirac measure for $\mu$-a.e. $\mathbf x$. Then, for each $\mathbf x$, there exists a point $\mathbf y_{\mathbf x}\in\mathcal X$ such that $K^n(\mathbf{x},\cdot)=\delta_{\mathbf y_{\mathbf x}}$, therefore, 
\begin{equation}
m_n(\mathbf x) \le \int_{\mathcal X} d_{\mathcal X}(\mathbf y_{\mathbf x},\mathbf u)\, \delta_{\mathbf y_{\mathbf x}}(\mathrm d\mathbf u)=0.
\end{equation}
Since $m_n(\mathbf x)\ge0$, we have $m_n(\mathbf x)=0$ for $\mu$-a.e. $\mathbf{x}$, thus $\mathcal E_n^* =\int_{\mathcal X}m_n(\mathbf x)\,\mathrm d\mu(\mathbf x)=0$.
\end{proof}

Proposition~\ref{prop:intrinsic_stochasticity} shows that $\mathcal{E}_n^* = 0$ serves as a certificate of $n$-step deterministic closure. However, $\mathcal{E}^{*}_{n}$ is not a certificate of non-injectivity. For instance, in the Lorenz63 system, the $z_3$-coordinate is invariant under the symmetry $(z_1,z_2,z_3) \mapsto (-z_1, -z_2, z_3)$ such that the time-delay map $D_{z_3, \tau, m}$ inherits this symmetry as a two-to-one non-injective map~\cite{duan2025causal}. In this case, the $n$-step future kernel can be Dirac while the time-delay map is non-injective.

\subsection{Geometric mechanism of information loss}\label{sec:geom_mech}

Here, we are concerned with the geometric mechanism that governs the occurrence of the non-deterministic closure. When the reconstructed system admits finite future branches, $\mathcal{E}_n^*$ can be decomposed into branch probabilities and future-branch separations, which shows that the local balance between the \textit{dynamical separation} generated by the $n$-step evolution and the \textit{observation curvature} of the reconstruction map $F$ determines the non-deterministic closure, rather than the prediction model itself. 

\subsubsection{Decomposition of folding}

In this work, we assume that the invariant measure $\nu$ admits SRB conditional measures from the disintegration of $\nu$ along the local unstable plaques. In particular, for a local unstable plaque $U$, the conditional measure $\nu^u_U$ is absolutely continuous with respect to the induced Riemannian volume $\mathcal{V}_U$: 
\begin{equation}
d\nu_{U}^u(\mathbf z)
=
\rho_{U}^u(\mathbf z)\,d\mathcal{V}_U(\mathbf z),
\end{equation}
where $\rho_U^u$ denotes the density on $U$. The local branch decomposition is illustrated in Figure~\ref{fig:local space decomposition}, which is a schematic two-dimensional cross-section of the local plaque and fibre geometry.

\begin{figure}\label{Latent Space Evolution}
    \centering
    \includegraphics[width=1\linewidth]{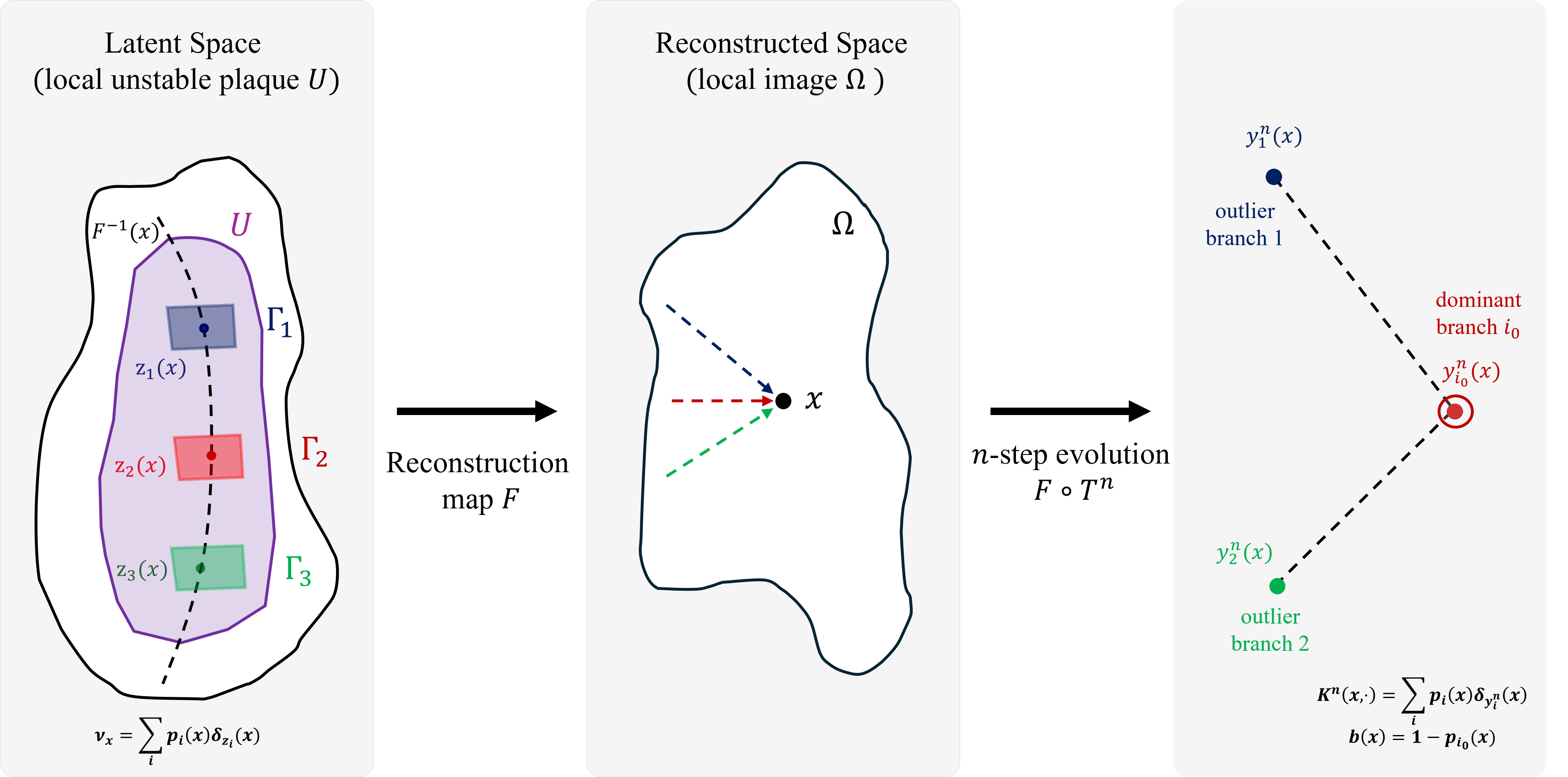}
    \caption{Local branch decomposition and future-branch separation. A local unstable plaque $U \subset \mathcal{M}$ contains local branch patches $\Gamma_i = \{\mathbf{z}_i(\mathbf{x}):\mathbf{x} \in \Omega\}$. For a given $\mathbf{x}$, the fibre $F^{-1}(\mathbf{x})$ intersects these branches at finitely many preimages $\mathbf{z}_i(\mathbf{x})$ with conditional weights $p_i(\mathbf{x})$, inducing the multimodal future kernel $K^n(\mathbf{x}, \cdot)$. Under the $n$-step pushforward, these atoms are mapped to $\mathbf{y}_i^n(\mathbf{x}) = F(T^n\mathbf{z}_i(\mathbf{x}))$.}
    \label{fig:local space decomposition}
\end{figure}

Fixing $\mathbf{x}$, we assume that the fibre $F^{-1}(\mathbf{x})$ contains finite distinct latent states $\{\mathbf{z}_{i}\}_{i=1}^{N}$ such that $F^{-1}(\mathbf{x})=\{\mathbf z_1(\mathbf{x}),\ldots,\mathbf z_N(\mathbf{x})\}$, where $\mathbf{z}_i(\mathbf{x})$ is the $i$-th latent preimage of $\mathbf{x}$ under $F$. Then the conditional measure $\nu_{\mathbf{x}}$ can be written as a finite atomic measure:
\begin{equation}
\nu_\mathbf{x} = \sum_{i=1}^N p_i(\mathbf{x}) \delta_{\mathbf{z}_i(\mathbf{x})},
\quad
p_i(\mathbf{x})\ge 0,\quad
\sum_{i=1}^N p_i(\mathbf{x})=1,
\end{equation}
which induces a \textit{multimodal future kernel}. Pushing $\nu_{\mathbf{x}}$ forward by $F\circ T^n$, the induced $n$-step future kernel splits into discrete branches:
\begin{equation}\label{eq:Kn_discrete}
K^n(\mathbf{x},\cdot) = \sum_{i=1}^N p_i(\mathbf{x})\,\delta_{\mathbf{y}_i^n(\mathbf{x})}, \qquad \mathbf{y}_i^n(\mathbf{x}) := F(T^n \mathbf{z}_i(\mathbf{x})).
\end{equation}

Now we focus on a local unstable plaque $U$ and fix a local image set $\Omega\subset\mathcal{X}$ with induced Riemannian volume $\mathcal{V}_{\Omega}$, such that, for every $\mathbf{x}\in\Omega$, the fibre $F^{-1}(\mathbf{x})\cap U$ consists of $N$ points $\{\mathbf z_1(\mathbf x),\ldots,\mathbf z_N(\mathbf x)\}$. For each $i=1,\ldots,N$, denote $\Gamma_i := \{\mathbf z_i(\mathbf x):\mathbf x\in\Omega\}\subset U$, and $F_i:=F|_{\Gamma_i}$. We assume that each $\Gamma_i$ is a full-dimensional $C^1$ submanifold of $U$, and $F_i:\Gamma_i\to\Omega$ is a $C^1$ diffeomorphism. Thus, the inverse branch $\mathbf z_i(\mathbf x)=F_i^{-1}(\mathbf x)$ is well-defined for every $\mathbf x\in\Omega$. 

Let $\nu_i:=\nu_U^u\!\restriction_{\Gamma_i}$ denote the unnormalized restriction of the unstable conditional measure to $\Gamma_i$ and $\mu_i:=(F_i)\#\nu_i$ denote the corresponding pushforward measure. Since $\nu_U^u$ is absolutely continuous with respect
to the induced Riemannian volume on \(U\), we have
\begin{equation}
d\nu_i(\mathbf z) = \rho_i^u(\mathbf z)\,d\mathcal V_{\Gamma_i}(\mathbf z),
\end{equation}
where $\mathcal V_{\Gamma_i}$ is the induced Riemannian volume on $\Gamma_i$ and $\rho_i^u$ denotes the restriction of $\rho_U^u$ to $\Gamma_i$. By the change of variables for $F_i$, 
\begin{equation}
\frac{d\mu_i}{d\mathcal V_\Omega}(\mathbf x)=\frac{\rho_i^u(\mathbf z_i(\mathbf x))}{J_{\Gamma_i}F(\mathbf z_i(\mathbf x))},
\end{equation}
where \(\mathbf z_i(\mathbf x)=F_i^{-1}(\mathbf x)\), and
$J_{\Gamma_i}F(\mathbf z):=\sqrt{\det\left[\left(D_{\Gamma_i}F(\mathbf z)\right)^*D_{\Gamma_i}F(\mathbf z)\right]}$ is the tangential Jacobian of $F$ along $\Gamma_i$. Define the local measure induced by $N$ branches as $\mu_\Omega:=\sum_{j=1}^N\mu_j$. Therefore, the conditional probability of the $i$-th branch at $\mathbf{x}\in\Omega$ is
\begin{equation}
p_i(\mathbf x) =
\frac{d\mu_i}{d\mu_\Omega}(\mathbf x)
=
\frac{
\dfrac{d\mu_i}{d\mathcal V_\Omega}(\mathbf x)
}{
\sum_{j=1}^N
\dfrac{d\mu_j}{d\mathcal V_\Omega}(\mathbf x)
} 
= 
\frac{
\rho_i^u(\mathbf z_i(\mathbf x))/
J_{\Gamma_i}F(\mathbf z_i(\mathbf x))
}{
\sum_{j=1}^N
\rho_j^u(\mathbf z_j(\mathbf x))/
J_{\Gamma_j}F(\mathbf z_j(\mathbf x))
}.
\end{equation}
Defining $w_i(\mathbf x):=\frac{\rho_i^u(\mathbf z_i(\mathbf x))}{J_{\Gamma_i}F(\mathbf z_i(\mathbf x))}$, we obtain
\begin{equation}\label{eq:weight_scaling}
p_i(\mathbf x) = \frac{w_i(\mathbf x)}{\sum_{j=1}^N w_j(\mathbf x)}.
\end{equation}

Equation \eqref{eq:weight_scaling} implies that the branch probability is determined by the relative local weights $w_i(\mathbf x)$. Thus, a larger branch weight may arise from a higher unstable density $\rho_i^u$, a smaller tangential area Jacobian $J_{\Gamma_i}F$, or both. In geometric terms, this reflects either more frequent visits of the system to the corresponding branch, stronger compression under the reconstruction map $F$, or both, when the $n$-step future evolution is non-deterministic. Next, we show that the competition between dynamical stretching and the curvature penalty of $F$ determines the extent to which these branches drift apart. 

\subsubsection{Competition: Stretching vs. Curvature}

To quantify the separation between competing future branches, fix $\mathbf{x}\in\Omega$ such that $F^{-1}(\mathbf{x})$ has finitely many local branches and let 
\begin{equation}
    i_0(\mathbf{x})\in\displaystyle\operatorname*{argmax}_{i\in\{1,\cdots,N\}}p_i(\mathbf{x})
\end{equation}
be a dominant branch, the one with maximal conditional probability $p_i(\mathbf{x})$. For each $i$, define its $n$-step future as $\boldsymbol{\xi}_i^n(\mathbf x):=T^n\mathbf z_i(\mathbf x)$. For each competing branch $j\neq i_0$, define the displacement vector between the two $n$-step future states by
\begin{equation}
\boldsymbol{\eta}_{n,j}(\mathbf x)
:=
\boldsymbol{\xi}_j^n(\mathbf x)
-
\boldsymbol{\xi}_{i_0}^n(\mathbf x),
\end{equation}
with length denoted by $r_{n,j}(\mathbf x):=\|\boldsymbol{\eta}_{n,j}(\mathbf x)\|$. The pairwise separation between the dominant branch and the competing branch \(j\) is
\begin{equation}
\Delta_{n,j}(\mathbf x)
:=
\left\|
F\!\left(\boldsymbol{\xi}_j^n(\mathbf x)\right)
-
F\!\left(\boldsymbol{\xi}_{i_0}^n(\mathbf x)\right)
\right\|_{\mathcal X}.
\end{equation}
The nearest competing future-branch separation is then
\begin{equation}\label{eq: nearest_competing_future_branch}
\Delta_n^-(\mathbf x)
:=
\min_{j\neq i_0}
\Delta_{n,j}(\mathbf x).
\end{equation}

Assume that $F$ is $C^2$ on the neighborhood containing the future branch points considered above. Along the relevant branch directions, define 
\begin{equation}\label{eq:alpha_F}
\alpha_{F,n}(\mathbf{x}):= \min_{j\neq i_0}\frac{\left\|\mathrm DF\!\left(\boldsymbol{\xi}_{i_0}^n(\mathbf x)\right)\boldsymbol{\eta}_{n,j}(\mathbf x)\right\|_{\mathcal X}}{r_{n,j}(\mathbf x)},
\end{equation}
where branches with $r_{n,j}(\mathbf x)=0$ are excluded from the minimum. Thus, $\alpha_{F,n}$ quantifies the smallest first-order separation of $F$ along the competing branch directions, and we assume that the second-order variation of $F$ along these branch directions is bounded by 
\begin{equation}~\label{eq:K_Hess}
K_{\mathrm{Hess}}(\mathbf{x})
:=
\sup_{j\neq i_0}
\sup_{0\le \theta\le 1}
\left\|
\mathrm D^2F
\left(
\boldsymbol{\xi}_{i_0}^n(\mathbf x)
+
\theta\boldsymbol{\eta}_{n,j}(\mathbf x)
\right)
\right\|_{\mathrm{op}},
\end{equation}
where $\|\mathrm D^2F(\xi)\|_{\mathrm{op}}$ denotes the operator norm of the Hessian viewed as a bilinear map. Then, taking the Taylor expansion of $F$ around $\boldsymbol{\xi}_{i_0}^n(\mathbf x)$, we have
\begin{equation}
F\!\left(\boldsymbol{\xi}_j^n(\mathbf x)\right)
-
F\!\left(\boldsymbol{\xi}_{i_0}^n(\mathbf x)\right)
=
\mathrm DF\!\left(\boldsymbol{\xi}_{i_0}^n(\mathbf x)\right)
\boldsymbol{\eta}_{n,j}(\mathbf x)
+
R_{n,j}(\mathbf x),
\end{equation}
where the Taylor remainder satisfies
\begin{equation}
\left\|
R_{n,j}(\mathbf x)
\right\|_{\mathcal X}
\le
\frac{K_{\mathrm{Hess}}(\mathbf{x})}{2}
r_{n,j}(\mathbf x)^2.
\end{equation}
Using the reverse triangle inequality, we have
\begin{equation}
\Delta_{n,j}(\mathbf x)
\ge
\left\|
\mathrm DF\!\left(\boldsymbol{\xi}_{i_0}^n(\mathbf x)\right)
\boldsymbol{\eta}_{n,j}(\mathbf x)
\right\|_{\mathcal X}
-
\left\|
R_{n,j}(\mathbf x)
\right\|_{\mathcal X}.
\end{equation}
By the definitions~\eqref{eq:alpha_F} and~\eqref{eq:K_Hess}, we obtain the pairwise lower bound
\begin{equation}
\Delta_{n,j}(\mathbf x) \ge \alpha_{F,n}(\mathbf x) r_{n,j}(\mathbf x) - \frac{K_{\mathrm{Hess}}(\mathbf x)}{2} r_{n,j}(\mathbf x)^2.
\end{equation}
Taking the minimum over all competing branches $j\neq i_0$ gives
\begin{equation}\label{eq:branch_separation_bound}
\Delta_n^-(\mathbf x)
\ge
\min_{j\neq i_0}
\left\{\alpha_{F,n}(\mathbf{x}) r_{n,j}(\mathbf x)-\frac{K_{\mathrm{Hess}}(\mathbf x)}{2}r_{n,j}(\mathbf x)^2\right\}.
\end{equation}

Equation~\eqref{eq:branch_separation_bound} shows that the local separation between the dominant future branch and a competing future branch is governed by the competition of two geometric effects. The factor $r_{n,j}(\mathbf{x})$ captures the dynamical separation generated by $T^n$, whereas $\alpha_{F,n}(\mathbf{x})$ measures the first-order visibility of this separation through $F$. Thus, $\alpha_{F,n}(\mathbf{x})r_{n,j}(\mathbf{x})$ is the first-order observed branch separation, while the term $(K_{\mathrm{Hess}}(\mathbf x)/2)r_{n,j}(\mathbf x)^2$ gives the second-order curvature penalty.


\subsubsection{The Scaling Law of Information Loss}
Combining the branch-weight decomposition in Eq.~\eqref{eq:weight_scaling} and the future-branch separation in Eq.~\eqref{eq:branch_separation_bound}, we now relate the pointwise risk $m_n(\mathbf{x})$, and hence the intrinsic stochasticity $\mathcal{E}^{*}_{n}$, to the product of branch mass and separation. For fixed $\mathbf{x}\in\Omega$ with the dominant branch $i_0=i_0(\mathbf x)$, we define the outlier mass by:
\begin{equation}\label{outlier mass}
b(\mathbf{x}) := 1 - p_{i_0}(\mathbf{x}) = \sum_{j\neq i_0}p_j(\mathbf x).
\end{equation}
For $b(\mathbf{x}) >0$, define the weighted dominant-to-outlier separation by:
\begin{equation}\label{eq:weighted_separation}
\overline{\Delta}_n(\mathbf x) := \frac{1}{b(\mathbf x)}\sum_{j\neq i_0}p_j(\mathbf x)d_{\mathcal X}\left(\mathbf y_{i_0}^n(\mathbf x),\mathbf y_j^n(\mathbf x)\right).
\end{equation}
If $b(\mathbf x)=0$, set $\overline{\Delta}_n(\mathbf x):=0$. Consistent with the notation in Eq.~\eqref{eq: nearest_competing_future_branch}, define the nearest and largest dominant-to-outlier separations:
\begin{equation}
\Delta_n^-(\mathbf x):=\min_{j\neq i_0} d_{\mathcal X} \left( \mathbf y_{i_0}^n(\mathbf x), \mathbf y_j^n(\mathbf x) \right), \quad \text{and} \quad \Delta_n^+(\mathbf x) :=\max_{j\neq i_0} d_{\mathcal X}\left(\mathbf y_{i_0}^n(\mathbf x),\mathbf y_j^n(\mathbf x)\right).
\end{equation}

\begin{proposition}[Dominant-branch mass-separation]
\label{prop:dominant_mass_separation}
Assume that the discrete future kernel admits the form:
\begin{equation}
K^n(\mathbf x,\cdot) =
\sum_{i=1}^{N} p_i(\mathbf x)\, \delta_{\mathbf y_i^n(\mathbf x)},
\qquad
\sum_{i=1}^{N}p_i(\mathbf x)=1.
\end{equation}
Suppose the dominant branch $i_0$ satisfies:
\begin{equation}\label{eq:branch condition}
p_{i_0}(\mathbf x)\ge b(\mathbf x)=1-p_{i_0}(\mathbf x), \quad \text{or equivalently,} \quad p_{i_0}(\mathbf x)\ge 1/2.
\end{equation}
Then
\begin{equation}\label{eq:pointwise_identity}
m_n(\mathbf x)
=
b(\mathbf x)\overline{\Delta}_n(\mathbf x).
\end{equation}
In particular,
\begin{equation}\label{eq:pointwise_sandwich}
b(\mathbf x)\Delta_n^-(\mathbf x)
\le
m_n(\mathbf x)
\le
b(\mathbf x)\Delta_n^+(\mathbf x).
\end{equation}
\end{proposition}

\begin{proof}
Let $\mathbf y_0 := \mathbf y_{i_0}^n(\mathbf x)$. For any \(\mathbf y\in\mathcal X\), write
\begin{equation}
a
:=
d_{\mathcal X}(\mathbf y,\mathbf y_0),
\qquad
b_j
:=
d_{\mathcal X}(\mathbf y,\mathbf y_j^n(\mathbf x)),
\qquad
D_j
:=
d_{\mathcal X}(\mathbf y_0,\mathbf y_j^n(\mathbf x)).
\end{equation}
Then
\begin{equation}
\ell_n(\mathbf x,\mathbf y)
=
p_{i_0}(\mathbf x)a
+
\sum_{j\neq i_0}
p_j(\mathbf x)b_j.
\end{equation}
Since \(p_{i_0}(\mathbf x)\ge b(\mathbf x)=\sum_{j\neq i_0}p_j(\mathbf x)\),
we have
\begin{equation}
p_{i_0}(\mathbf x)a
\ge
\sum_{j\neq i_0}p_j(\mathbf x)a.
\end{equation}
Therefore,
\begin{equation}
\ell_n(\mathbf x,\mathbf y)
\ge
\sum_{j\neq i_0}
p_j(\mathbf x)
\left(
a+b_j
\right).
\end{equation}
By the triangle inequality $a+b_j \ge D_j$, we have
\begin{equation}
\ell_n(\mathbf x,\mathbf y) \ge \sum_{j\neq i_0} p_j(\mathbf x) d_{\mathcal X}\left(\mathbf y_{i_0}^n(\mathbf x),\mathbf y_j^n(\mathbf x)\right) = b(\mathbf x)\overline{\Delta}_n(\mathbf x).
\end{equation}
Taking the infimum over $\mathbf y\in\mathcal X$ gives
\begin{equation}
m_n(\mathbf x)\ge b(\mathbf x)\overline{\Delta}_n(\mathbf x).
\end{equation}
On the other hand, taking pointwise risk at $\mathbf y=\mathbf y_{i_0}^n(\mathbf x)$ yields
\begin{equation}
m_n(\mathbf x) \le \ell_n(\mathbf x,\mathbf y_{i_0}^n(\mathbf x)) = \sum_{j\neq i_0} p_j(\mathbf x) d_{\mathcal X} \left( \mathbf y_{i_0}^n(\mathbf x), \mathbf y_j^n(\mathbf x) \right) = b(\mathbf x)\overline{\Delta}_n(\mathbf x).
\end{equation}
Thus,
\begin{equation}
m_n(\mathbf x) = b(\mathbf x)\overline{\Delta}_n(\mathbf x),
\end{equation}
and 
\begin{equation}
\Delta_n^-(\mathbf x)
\le
\overline{\Delta}_n(\mathbf x)
\le
\Delta_n^+(\mathbf x).
\end{equation}
\end{proof}

If the outlier separations are comparable on the local region $A_n$ that satisfies Eq.~\eqref{eq:branch condition}, namely if there exists a constant \(C_\Delta\ge1\) such that 
\begin{equation}
\Delta_n^+(\mathbf x)
\le
C_\Delta\Delta_n^-(\mathbf x),
\qquad
\mathbf x\in A_n,
\end{equation}
then
\begin{equation}\label{eq:scaling_law}
\mathcal E_{n,A_n}^*
\asymp
\int_{A_n}
b(\mathbf x)\Delta_n^-(\mathbf x)\,
\mathrm d\mu(\mathbf x),
\end{equation}
with constants depending only on \(C_\Delta\). Finally, combining this with Eq.~\eqref{eq:branch_separation_bound}, we obtain the explicit mass-separation mechanism
\begin{equation}
\mathcal E_{n,A_n}^* \ge \int_{A_n}
b(\mathbf x)
\left[
\min_{j\neq i_0}
\left\{
\alpha_{F,n}(\mathbf{x})r_{n,j}(\mathbf x) -
\frac{K_{\mathrm{Hess}}(\mathbf{x})}{2}
r_{n,j}(\mathbf x)^2
\right\}
\right]_+
\,\mathrm d\mu(\mathbf x),
\end{equation}
where $[s]_+:=\max\{s,0\}$. Thus, within this dominant-branch regime, a nontrivial contribution to $\mathcal{E}^*$ requires both non-negligible outlier mass and a separated future branch.

\section{Estimation and numerical experiments}\label{Numerical Verification}

From Section~\ref{Section three}, the essential point for checking deterministic closure of the evolution is estimating the dispersion of the $n$-step pushforward kernel $K^n(\mathbf{x},\cdot)$. In this section, we use the empirical $k$-NN kernel to approximate $K^n(\mathbf{x},\cdot)$ in the reconstructed space, and compute the corresponding Fr\'echet median via a Weiszfeld iteration~\cite{weiszfeld1937point}.

\subsection{Estimator and algorithm}

Fix a pushforward horizon $n \ge 1$, and take $q\in\{0,1,\cdots,T-n-1\}$. For each $\mathbf{x}_q$, let $\mathcal{N}_k(\mathbf{x}_q)$ denote the set of the $k$ nearest neighbors of $\mathbf{x}_q$, after applying a Theiler window, and the $k$-NN radius is defined as: 
\begin{equation}\label{kNN_radius}
\begin{aligned}
    r_{k,N}(\mathbf{x}_q) := \max_{t\in \mathcal{N}_k(q)}d_\mathcal{X}(\mathbf{x}_t,\mathbf{x}_q).
\end{aligned}
\end{equation}
Pushing the neighborhood cloud forward by $n$ steps yields a cloud $\mathcal{Y}_q := \{\mathbf{y}_j\}_{j=1}^k := \{\mathbf{x}_{t+n}: t\in\mathcal{N}_k(\mathbf{x}_q)\}\subset \mathcal{X}$, which approximates the conditional kernel $K^n(\mathbf{x}_q,\cdot)$ at the finite resolution $r_{k,N}(\mathbf{x}_q)$. Throughout this section, we use the uniform empirical kernel:
\begin{equation}\label{eq:knn_uniform_kernel}
\begin{aligned}
    \widehat K_{N,k}^n(\mathbf{x}_q,\cdot)\ :=\ \frac{1}{k}\sum_{j=1}^k \delta_{\mathbf{y}_j}.
\end{aligned}
\end{equation}
Given $\widehat K_{N,k}^n(\mathbf{x}_q,\cdot)$, the empirical pointwise risk is:
\begin{equation}\label{eq:ellhat_nk}
\widehat{\ell}_{n,k,N}(\mathbf{x}_q,\mathbf{y})
\ :=\ 
\int_{\mathcal{X}} d_{\mathcal{X}}(\mathbf{y},\mathbf{y}')\,\widehat K_{N,k}^n(\mathbf{x}_q,\mathrm{d}\mathbf{y}')
\ =\ 
\frac{1}{k}\sum_{j=1}^k d_{\mathcal{X}}(\mathbf{y},\mathbf{y}_j),
\end{equation}
and the empirical Fr\'echet median cost is
\begin{equation}\label{eq:mhat_nk}
\widehat m_{n,k,N}(\mathbf{x}_q)\ :=\ \min_{\mathbf{y}\in \mathcal{X}}\widehat{\ell}_{n,k,N}(\mathbf{x}_q,\mathbf{y})
\ =\ 
\min_{\mathbf{y}\in \mathcal{X}}\frac{1}{k}\sum_{j=1}^k d_{\mathcal{X}}(\mathbf{y},\mathbf{y}_j).
\end{equation}
Moreover, for a fixed query set $\mathcal{Q}$, the numerical approximation of $\mathcal{E}_{n}^{*}$ under finite sample and finite resolution is:
\begin{equation}\label{eq:Estar_nk}
\widehat{\mathcal{E}}_{n,k,N}^{*}
\ :=\ 
\frac{1}{|\mathcal{Q}|}\sum_{q\in\mathcal{Q}}\widehat m_{n,k,N}(\mathbf{x}_q).
\end{equation}
The following Lemma~\ref{lem:best_dirac_lipschitz} is conditional: if an empirical kernel $\widehat K^n$ converges to $K^n$ in $\mathcal{W}_1$, then the corresponding best-Dirac costs converge and $\widehat m_{n,k,N}(\mathbf x)$ converges to $m_n(\mathbf x)$.   

\begin{lemma}\label{lem:best_dirac_lipschitz}
For $\eta \in \mathcal{P}(F(\mathcal{M}))$, define
\begin{equation}
\begin{aligned}
\mathcal{J}(\eta)
:= \inf_{\mathbf{y}\in \mathcal{X}}\int_{\mathcal{X}} d_{\mathcal{X}}(\mathbf{y},\xi)\,\eta(d\xi) = \inf_{\mathbf{y}\in \mathcal{X}} \mathcal{W}_1(\eta,\delta_\mathbf{y}).
\end{aligned}
\end{equation}
Then for any $\eta,\zeta\in \mathcal{P}(F(\mathcal{M}))$,
\begin{equation}
\begin{aligned}
\big|\mathcal{J}(\eta)-\mathcal{J}(\zeta)\big| \le \mathcal{W}_1(\eta,\zeta).
\end{aligned}
\end{equation}
If $\mathcal{W}_1\left(\widehat K^n_{N,k}(\mathbf x,\cdot),
K^n(\mathbf x,\cdot)\right)\to0$, then $\widehat m_{n,k,N}(\mathbf x)\to m_n(\mathbf x)$, $\text{for $\mu$-a.e. }\mathbf{x}$.   
\end{lemma}

\begin{proof}
For any fixed $\mathbf{y}\in X$, the triangle inequality for $\mathcal{W}_1$ gives
\begin{equation}
\begin{aligned}
\mathcal{W}_1(\eta,\delta_\mathbf{y}) \le \mathcal{W}_1(\eta,\zeta)+\mathcal{W}_1(\zeta,\delta_\mathbf{y}).
\end{aligned}
\end{equation}
Taking the infimum over $\mathbf{y}$ yields
\begin{equation}
\begin{aligned}
\mathcal{J}(\eta) \le \mathcal{W}_1(\eta,\zeta)+\mathcal{J}(\zeta).
\end{aligned}
\end{equation}
Exchanging the roles of $\eta$ and $\zeta$ gives
\begin{equation}
\begin{aligned}
\mathcal{J}(\zeta) \le \mathcal{W}_1(\eta,\zeta)+\mathcal{J}(\eta).
\end{aligned}
\end{equation}
Combining the two inequalities proves
\begin{equation}
\begin{aligned}
\big|\mathcal{J}(\eta)-\mathcal{J}(\zeta)\big| \le \mathcal{W}_1(\eta,\zeta).
\end{aligned}
\end{equation}
Since $\mathcal{M}$ is compact and $F$ is continuous, $F(\mathcal{M})$ is compact, thus weak convergence of probability measures implies $\mathcal{W}_1$-convergence. Thus, the convergence of $\widehat K_{N,k}^n(\mathbf{x},\cdot)$ implies $\widehat m_{n,k,N}(\mathbf{x})\to m_n(\mathbf{x})$.
\end{proof}

The proposed algorithm is outlined in Algorithm~\ref{alg:e_star_fixed_k}. 

\begin{algorithm}[!ht]
\caption{Estimating intrinsic stochasticity at fixed neighborhood scale $k$}
\label{alg:e_star_fixed_k}
\begin{algorithmic}[1]
\Require Reconstructed states $\{\mathbf{x}_t\}_{t=0}^{T-1}\subset \mathcal{X}$; horizon $n$; neighborhood count $k$; Theiler window $w$; Number of queries $N_\mathcal{Q}$.
\Ensure $\widehat{\mathcal{E}}_{n,k,N}^*$.
\State Sample query indices $\mathcal{Q}\subset\{0,\dots,T-n-1\}$ uniformly at random with $|\mathcal{Q}|=N_{\mathcal{Q}}$.
\For{$q\in\mathcal{Q}$}
    \State Find $\mathcal{N}_k(\mathbf{x}_q)$, the set of $k$ nearest neighbors $\mathcal{N}_{k}(\mathbf{x}_q)$ of $\mathbf{x}_q$ with the Theiler window. 
    \State Form $\mathcal{Y}_q:=\{\mathbf{x}_{t+n}:t\in\mathcal{N}_k(\mathbf{x}_q)\}$.
    \State Compute $\widehat m_{n,k,N}(\mathbf{x}_q)=\min_{\mathbf{y}\in X}\frac{1}{k}\sum_{\mathbf{y}_j\in\mathcal{Y}_q} d_{\mathcal{X}}(\mathbf{y},\mathbf{y}_j)$ via Weiszfeld iteration.
\EndFor
\State \Return $\widehat{\mathcal{E}}_{n,k,N}^{*}=\frac{1}{|\mathcal{Q}|}\sum_{q\in\mathcal{Q}}\widehat m_{n,k,N}(\mathbf{x}_q)$.
\end{algorithmic}
\end{algorithm}

\subsection{Finite resolution bounds} 
The $k$-NN-based estimator $\widehat{\mathcal{E}}_{n,k,N}^*$ is computed at a finite spatial resolution $r_{k,N}(\mathbf{x}_q)$. The following bounds separate the finite-resolution single-branch contribution from the possible multi-branch contribution. 

\begin{proposition}[Upper bound on a single isolated branch]
\label{prop:single_branch_upper}
Fix a horizon $n\ge 1$ and a query index $q$. Assume there exists an open set
$U\subset \mathcal{M}$ and an open set $V\subset \mathcal{X}$ such that
\begin{enumerate}
    \item $\mathbf{x}_q\in V\cap F(\mathcal{M})$;
    \item $F^{-1}(V\cap F(\mathcal{M}))=U$;
    \item $F|_U:U\to V\cap F(\mathcal{M})$ is a $C^1$ diffeomorphism.
\end{enumerate}
Let
\begin{equation}
\begin{aligned}
G := (F|_U)^{-1}:V\cap F(\mathcal{M})\to U, \quad Y_n := F\circ T^n\circ G:V\cap F(\mathcal{M})\to \mathcal{X}.
\end{aligned}
\end{equation}
Assume that $Y_n$ is $L_n(q)$-Lipschitz on $V\cap F(\mathcal{M})$. If
\begin{equation}
\begin{aligned}
B_X\!\big(\mathbf{x}_q,r_{k,N}(\mathbf{x}_q)\big)\cap F(\mathcal{M})\subset V,
\end{aligned}
\end{equation}
then
\begin{equation}
\begin{aligned}
\widehat m_{n,k,N}(\mathbf{x}_q) \le L_n(q)\,r_{k,N}(\mathbf{x}_q).
\end{aligned}
\end{equation}
In particular, for fixed $n$, if $r_{k,N}(\mathbf{x}_q)\to 0$, then
\begin{equation}
\begin{aligned}
\widehat m_{n,k,N}(\mathbf{x}_q) \to 0.
\end{aligned}
\end{equation}
\end{proposition}

\begin{proof}
Write
\begin{equation}
\begin{aligned}
\mathcal{N}_k(q) = \{t_1,\dots,t_k\}, \qquad \mathbf{y}_j := \mathbf{x}_{t_j+n}, \qquad j=1,\dots,k.
\end{aligned}
\end{equation}
By the assumption
\begin{equation}
\begin{aligned}
B_{\mathcal{X}}\!\big(\mathbf{x}_q,r_{k,N}(\mathbf{x}_q)\big)\cap F(\mathcal{M})\subset V,
\end{aligned}
\end{equation}
every neighbor $\mathbf{x}_{t_j}$ lies in $V\cap F(\mathcal{M})$. Since
$F^{-1}(V\cap F(\mathcal{M}))=U$ and $G=(F|_U)^{-1}$, the latent state generating $\mathbf{x}_{t_j}$
must be
\begin{equation}
\begin{aligned}
\mathbf{z}_{t_j} = G(\mathbf{x}_{t_j})\in U,
\end{aligned}
\end{equation}
and hence
\begin{equation}
\begin{aligned}
\mathbf{y}_j &= \mathbf{x}_{t_j+n} = F(T^n \mathbf{z}_{t_j}) = F(T^n G(\mathbf{x}_{t_j})) = Y_n(\mathbf{x}_{t_j}).
\end{aligned}
\end{equation}
Now choose the test point
\begin{equation}
\begin{aligned}
y_* := Y_n(\mathbf{x}_q) = F(T^n G(\mathbf{x}_q)).
\end{aligned}
\end{equation}
By definition of $\widehat m_{n,k,N}(\mathbf{x}_q)$,
\begin{equation}
\begin{aligned}
\widehat m_{n,k,N}(\mathbf{x}_q) = \inf_{\mathbf{y}\in \mathcal{X}}\frac{1}{k}\sum_{j=1}^k d_{\mathcal{X}}(\mathbf{y},\mathbf{y}_j) \le \frac{1}{k}\sum_{j=1}^k d_{\mathcal{X}}(\mathbf{y}_*,\mathbf{y}_j).
\end{aligned}
\end{equation}
Using the Lipschitz property of $Y_n$,
\begin{equation}
\begin{aligned}
d_{\mathcal{X}}(\mathbf{y}_*,\mathbf{y}_j) = d_{\mathcal{X}}\!\big(Y_n(\mathbf{x}_q),Y_n(\mathbf{x}_{t_j})\big) \le L_n(q)\,d_{\mathcal{X}}(\mathbf{x}_q,\mathbf{x}_{t_j}) \le L_n(q)\,r_{k,N}(\mathbf{x}_q).
\end{aligned}
\end{equation}
Averaging over $j=1,\dots,k$ yields
\begin{equation}\label{eq: upper bound}
\begin{aligned}
\widehat m_{n,k,N}(\mathbf{x}_q) \le L_n(q)\,r_{k,N}(\mathbf{x}_q).
\end{aligned}
\end{equation}
The final claim follows immediately when $r_{k,N}(\mathbf{x}_q)\to 0$.
\end{proof}

Proposition~\ref{prop:single_branch_upper} shows that a large $\widehat m_{n,k,N}$ is not, by itself, evidence of non-deterministic closure. Even on a single isolated branch, the empirical cost can be large if the neighborhood radius $r_{k,N}$ or the Lipschitz constant is large. 

\begin{proposition}[Lower bound from two separated branches]
\label{prop:two_branch_lower}
Fix a horizon $n\ge 1$ and a query index $q$. Assume there exist two disjoint open sets $U_1,U_2\subset \mathcal{M}$ and an open set $V\subset \mathcal{X}$ such that $\mathbf{x}_q\in V\cap F(\mathcal{M})$ and,
for $i=1,2$,
\begin{equation}
\begin{aligned}
F|_{U_i}:U_i\to V\cap F(\mathcal{M})
\end{aligned}
\end{equation}
is a $C^1$ diffeomorphism. Let
\begin{equation}
\begin{aligned}
G_i &:= (F|_{U_i})^{-1}:V\cap F(\mathcal{M})\to U_i, \\
Y_{i,n} &:= F\circ T^n\circ G_i:V\cap F(\mathcal{M})\to \mathcal{X}.
\end{aligned}
\end{equation}
Assume that $Y_{i,n}$ is $L_{i,n}(q)$-Lipschitz on $V\cap F(\mathcal{M})$ for $i=1,2$, and define
\begin{equation}
\begin{aligned}
\Delta_n(\mathbf{x}_q) := d_{\mathcal{X}}\!\big(Y_{1,n}(\mathbf{x}_q),Y_{2,n}(\mathbf{x}_q)\big) 
= d_{\mathcal{X}}\!\big(\mathbf{y}_1^n(\mathbf{x}_q),\mathbf{y}_2^n(\mathbf{x}_q)\big) > 0.
\end{aligned}
\end{equation}
Write
\begin{equation}
\begin{aligned}
\mathcal{N}_k(q) &= \{t_1,\dots,t_k\}, \\
\mathcal{I}_i(q) &:= \{\ell\in\{1,\dots,k\}: \mathbf{z}_{t_\ell}\in U_i\}, \\
\widehat p_i(q) &:= \frac{|\mathcal{I}_i(q)|}{k}, \quad i=1,2.
\end{aligned}
\end{equation}
If $B_X\!\big(\mathbf{x}_q,r_{k,N}(\mathbf{x}_q)\big)\cap F(\mathcal{M})\subset V$, then
\begin{equation}
\begin{aligned}
\widehat m_{n,k,N}(\mathbf{x}_q) \ge \min\{\widehat p_1(q),\widehat p_2(q)\}\,\Delta_n(\mathbf{x}_q) - \Big(\widehat p_1(q)L_{1,n}(q)+\widehat p_2(q)L_{2,n}(q)\Big)\,r_{k,N}(\mathbf{x}_q).
\end{aligned}
\end{equation}
\end{proposition}

\begin{proof}
For each $\ell\in \mathcal{I}_i(q)$, the neighbor $\mathbf{x}_{t_\ell}$ is generated by a latent state
$\mathbf{z}_{t_\ell}\in U_i$, hence
\begin{equation}
\begin{aligned}
\mathbf{x}_{t_\ell} = F(\mathbf{z}_{t_\ell}),\quad \mathbf{y}_\ell = \mathbf{x}_{t_\ell+n} = F(T^n \mathbf{z}_{t_\ell}) = Y_{i,n}(\mathbf{x}_{t_\ell}).
\end{aligned}
\end{equation}
Therefore,
\begin{equation}
\begin{aligned}
d_{\mathcal{X}}\!\big(\mathbf{y}_\ell,Y_{i,n}(\mathbf{x}_q)\big) = d_{\mathcal{X}}\!\big(Y_{i,n}(\mathbf{x}_{t_\ell}),Y_{i,n}(\mathbf{x}_q)\big) \le L_{i,n}(q)\,d_{\mathcal{X}}(\mathbf{x}_{t_\ell},\mathbf{x}_q) \le L_{i,n}(q)\,r_{k,N}(\mathbf{x}_q).
\end{aligned}
\end{equation}
Now fix any $\mathbf{y}\in \mathcal{X}$. By the reverse triangle inequality, for each
$\ell\in\mathcal{I}_i(q)$,
\begin{equation}
\begin{aligned}
d_{\mathcal{X}}(\mathbf{y},\mathbf{y}_\ell) \ge d_{\mathcal{X}}\!\big(\mathbf{y},Y_{i,n}(\mathbf{x}_q)\big) - d_{\mathcal{X}}\!\big(\mathbf{y}_\ell,Y_{i,n}(\mathbf{x}_q)\big) 
\ge d_{\mathcal{X}}\!\big(\mathbf{y},Y_{i,n}(\mathbf{x}_q)\big) - L_{i,n}(q)\,r_{k,N}(\mathbf{x}_q).
\end{aligned}
\end{equation}
Hence
\begin{equation}
\begin{aligned}
\frac{1}{k}\sum_{\ell=1}^k d_{\mathcal{X}}(\mathbf{y},\mathbf{y}_\ell) \ge \sum_{i=1}^2 \frac{1}{k}\sum_{\ell\in \mathcal{I}_i(q)} \Big(d_{\mathcal{X}}\!\big(\mathbf{y},Y_{i,n}(\mathbf{x}_q)\big)-L_{i,n}(q)\,r_{k,N}(\mathbf{x}_q) \Big).
\end{aligned}
\end{equation}
Since $|\mathcal{I}_i(q)|/k=\widehat p_i(q)$, we have
\begin{equation}
\begin{aligned}
\frac{1}{k}\sum_{\ell=1}^k d_{\mathcal{X}}(\mathbf{y},\mathbf{y}_\ell) &\ge \widehat p_1(q)\,d_{\mathcal{X}}\!\big(\mathbf{y},Y_{1,n}(\mathbf{x}_q)\big) + \widehat p_2(q)\,d_{\mathcal{X}}\!\big(\mathbf{y},Y_{2,n}(\mathbf{x}_q)\big) \\
&\quad - \Big(\widehat p_1(q)L_{1,n}(q)+\widehat p_2(q)L_{2,n}(q)\Big)\,r_{k,N}(\mathbf{x}_q).
\end{aligned}
\end{equation}
Using
\begin{equation}
\begin{aligned}
\alpha a+\beta b \ge \min\{\alpha,\beta\}(a+b), \qquad \alpha,\beta,a,b\ge 0,
\end{aligned}
\end{equation}
and the triangle inequality
\begin{equation}
\begin{aligned}
d_{\mathcal{X}}\!\big(\mathbf{y},Y_{1,n}(\mathbf{x}_q)\big)+d_{\mathcal{X}}\!\big(\mathbf{y},Y_{2,n}(\mathbf{x}_q)\big) &\ge d_{\mathcal{X}}\!\big(Y_{1,n}(\mathbf{x}_q),Y_{2,n}(\mathbf{x}_q)\big) = \Delta_n(\mathbf{x}_q),
\end{aligned}
\end{equation}
we obtain
\begin{equation}
\begin{aligned}
\frac{1}{k}\sum_{\ell=1}^k d_{\mathcal{X}}(\mathbf{y},\mathbf{y}_\ell) \ge \min\{\widehat p_1(q),\widehat p_2(q)\}\,\Delta_n(\mathbf{x}_q) - \Big(\widehat p_1(q)L_{1,n}(q)+\widehat p_2(q)L_{2,n}(q)\Big)\,r_{k,N}(\mathbf{x}_q).
\end{aligned}
\end{equation}
Taking the infimum over $\mathbf{y}\in \mathcal{X}$ proves the claim.
\end{proof}

If additional branches are present, the same estimate applies to the two selected branches after excluding the nonnegative contribution of the other neighbors. 

\subsection{Synthetic examples}\label{Rossler Case}

We apply $\widehat{\mathcal E}_{n,k,N}^{*}$ as a finite-horizon, finite-resolution closure score on the R\"ossler system under several differential maps of the form in Eq.~\eqref{derivative coordinate map}. Our goal is to examine how the empirical closure score changes with the choice of reconstruction. Although the main motivation of this paper is the time-delay reconstruction, our theoretical framework in Section~\ref{Section three} is formulated for any fixed reconstruction map $F$. The differential maps used here are treated as controlled analytic reconstructions generated by the same observables. A large value of $\widehat{\mathcal E}_{n,k,N}^{*}$ may arise from multi-valued closure, from the radius of the $k$-NN cloud, or from expansion under the $n$-step evolution. 

We consider coordinate projections, linear combinations of coordinate projections, and differential embeddings involving multiple coordinate projections. Analytic expressions for all differential embeddings are provided in Appendix~\ref{ape: Rossler}. For all experiments, the pushforward step $n$ is fixed at $20$, with a sample size of $N=5000$ and $k=50$. At the population level, an injective reconstruction map gives a Dirac future kernel with a deterministic closure in the reconstructed space. At finite resolution, however, $\widehat{\mathcal E}_{n,k,N}^*$ may still be non-negligible because the $k$-NN cloud has a positive radius which can be amplified by the $n$-step dynamics. 

A comparison between $\widehat{\mathcal E}_{n,k,N}^*$ and three simpler finite-resolution baselines computed from the same pushforward neighbor clouds is shown in Table~\ref{tab:future_cloud_baseline_comparison}. The implementation details and theoretical relationships are given in Appendix~\ref{app:future_cloud_baselines}. The rankings are consistent across all diagnostics, which is expected, since $\widehat{\mathcal E}_{n,k,N}^*$ is itself a first-moment best-Dirac dispersion of the empirical future cloud. At the same time, conditional variance gives a similar qualitative ordering. However, the scale is dominated by squared distances and is therefore much more sensitive to the upper tails of the future cloud. Comparison shows that $\widehat{\mathcal E}_{n,k,N}^*$ is closely related to future-cloud dispersion, as expected from its definition as a first-moment best-Dirac cost. Rather, its advantage lies in the population-level interpretation that is closely linked to the zero-Dirac criterion in Proposition~\ref{prop:intrinsic_stochasticity} and the mass-separation mechanism in Proposition~\ref{prop:dominant_mass_separation}.   

\begin{table}[!ht]
\centering
\caption{Comparison between the finite-resolution estimator
$\widehat{\mathcal E}_{n,k,N}^{*}$ and three simple future-cloud baseline diagnostics for selected R\"ossler differential reconstructions. All quantities are computed from the same Theiler-windowed $k$-nearest-neighbor ($k$-NN) future clouds after standardization. Results are reported as mean $\pm$ standard deviation over $30$ experiments, with $n=20$, $k=50$, $N_{\mathcal Q}=5000$, and Theiler window $w = 30$. $\hat{V}_{n,k}$ is the empirical conditional variance trace, $\hat{P}_{n,k}$ is the local short-horizon prediction error, and $\hat{D}^{\mathrm{pair}}_{n,k}$ is the mean pairwise future-cloud dispersion. Reconstructions marked with ${}^{\ast}$ are diffeomorphisms. Lower values indicate a smaller finite-resolution future-cloud spread.}
\label{tab:future_cloud_baseline_comparison}
\setlength{\tabcolsep}{4.8pt}
\renewcommand{\arraystretch}{1.15}
\small
\begin{tabular}{lcccc}
\hline
\textbf{Reconstruction}
&
\(\widehat{\mathcal E}_{n,k,N}^*\)
&
$\hat{V}_{n,k}$
&
$\hat{P}_{n,k}$
&
$\hat{D}^{\mathrm{pair}}_{n,k}$
\\
\hline

\(\left(z_2,\dot z_2,\ddot z_2\right)^{\ast}\)
& \(0.05 \pm 0.00\)
& \(0.01 \pm 0.00\)
& \(0.04 \pm 0.00\)
& \(0.06 \pm 0.00\)
\\

\(\left(z_1,z_2,\dot z_1\right)^{\ast}\)
& \(0.05 \pm 0.00\)
& \(0.01 \pm 0.00\)
& \(0.04 \pm 0.00\)
& \(0.06 \pm 0.00\)
\\

\(\left(z_1+z_2,\dot z_1+\dot z_2,\ddot z_1+\ddot z_2\right)\)
& \(0.09 \pm 0.00\)
& \(0.12 \pm 0.01\)
& \(0.08 \pm 0.00\)
& \(0.12 \pm 0.00\)
\\

\(\left(z_1,\dot z_1,\ddot z_1\right)\)
& \(0.11 \pm 0.00\)
& \(0.21 \pm 0.02\)
& \(0.09 \pm 0.00\)
& \(0.15 \pm 0.01\)
\\

\(\left(z_2+z_3,\dot z_2+\dot z_3,\ddot z_2+\ddot z_3\right)\)
& \(0.72 \pm 0.04\)
& \(19.58 \pm 1.53\)
& \(0.57 \pm 0.03\)
& \(1.01 \pm 0.05\)
\\

\(\left(z_1+z_3,\dot z_1+\dot z_3,\ddot z_1+\ddot z_3\right)\)
& \(1.86 \pm 0.10\)
& \(134.67 \pm 10.31\)
& \(1.91 \pm 0.13\)
& \(2.56 \pm 0.13\)
\\

\(\left(z_1,z_3,\dot z_1\right)^{\ast}\)
& \(4.71 \pm 0.19\)
& \(387.63 \pm 19.11\)
& \(4.72 \pm 0.20\)
& \(5.64 \pm 0.22\)
\\

\(\left(z_2,z_3,\dot z_2\right)^{\ast}\)
& \(4.71 \pm 0.19\)
& \(387.91 \pm 19.17\)
& \(4.74 \pm 0.20\)
& \(5.65 \pm 0.22\)
\\

\(\left(z_2,z_3,\dot z_3\right)\)
& \(12.01 \pm 0.55\)
& \(5143.69 \pm 358.06\)
& \(12.56 \pm 0.77\)
& \(17.11 \pm 0.78\)
\\

\(\left(z_1,z_3,\dot z_3\right)\)
& \(12.06 \pm 0.55\)
& \(5144.71 \pm 358.08\)
& \(12.62 \pm 0.77\)
& \(17.17 \pm 0.78\)
\\

\(\left(z_3,\dot z_3,\ddot z_3\right)\)
& \(14.25 \pm 0.63\)
& \(5379.02 \pm 403.17\)
& \(12.03 \pm 0.64\)
& \(19.44 \pm 0.85\)
\\

\hline
\end{tabular}
\end{table}

Reconstructions involving $z_3$, such as $(z_2,z_3,\dot z_2)$ and $(z_1,z_3,\dot z_1)$, are diffeomorphisms but yield larger finite-$k$ values. However, this does not imply that they have non-deterministic closure at the population level. We emphasize that a large $\widehat{\mathcal{E}}_{n,k,N}^{*}$ alone is not a sign of multi-valued closure since it is also affected by the neighborhood radius and the local $n$-step expansion. To make this effect explicit, we summarize the $k$-NN radius $r_{k,N}(\mathbf{x}_q)$ by its median $\widetilde r_k:=\operatorname{median}_{q\in \mathcal{Q}}r_{k,N}(\mathbf{x}_q)$ and its 90\% quantile $\widetilde r_{k,0.9}:=\operatorname{quantile}_{0.9}\big(r_{k,N}(\mathbf{x}_q)\big)$. The results listed in Table~\ref{tab:rk_rk90_clean} suggest that the large values of $\widehat{\mathcal{E}}_{n,k,N}^{*}$ for the $z_3$-based reconstructions are at least partly driven by finite-resolution effects. In particular, the elevated upper tail indicates that a nontrivial fraction of query points require much larger neighborhoods to collect $k$ neighbors. Since the pushforward cloud is evaluated at a fixed horizon $n$, these large-radius neighborhoods can undergo stronger finite-time expansion and contribute disproportionately to the average Fr\'echet-median cost. 

\begin{table}[!ht]
\centering
\caption{Effective neighborhood radius diagnostics. We report the median $\widetilde r_k$ and the $90\%$ quantile $\widetilde r_{k,0.9}$ of the $k$-NN radius over query points.}
\label{tab:rk_rk90_clean}
\setlength{\tabcolsep}{10pt}
\renewcommand{\arraystretch}{1.15}
\begin{tabular}{lcc}
\hline
\textbf{Reconstruction} 
& \(\boldsymbol{\widetilde r_k}\) 
& \(\boldsymbol{\widetilde r_{k,0.9}}\) \\
\hline

\(\left(z_2,\dot z_2,\ddot z_2\right)\)
& \(0.07\) & \(0.17\) \\

\(\left(z_1,z_2,\dot z_1\right)\)
& \(0.07\) & \(0.17\) \\

\(\left(z_1+z_2,\dot z_1+\dot z_2,\ddot z_1+\ddot z_2\right)\)
& \(0.06\) & \(0.18\) \\

\(\left(z_1,\dot z_1,\ddot z_1\right)\)
& \(0.07\) & \(0.19\) \\

\(\left(z_2+z_3,\dot z_2+\dot z_3,\ddot z_2+\ddot z_3\right)\)
& \(0.06\) & \(0.39\) \\

\(\left(z_1+z_3,\dot z_1+\dot z_3,\ddot z_1+\ddot z_3\right)\)
& \(0.06\) & \(0.69\) \\

\(\left(z_1,z_3,\dot z_1\right)\)
& \(0.06\) & \(0.92\) \\

\(\left(z_2,z_3,\dot z_2\right)\)
& \(0.07\) & \(0.93\) \\

\(\left(z_2,z_3,\dot z_3\right)\)
& \(0.08\) & \(8.05\) \\

\(\left(z_1,z_3,\dot z_3\right)\)
& \(0.07\) & \(8.05\) \\

\(\left(z_3,\dot z_3,\ddot z_3\right)\)
& \(0.07\) & \(17.08\) \\

\hline
\end{tabular}
\end{table}

The characteristics of the $z_3$ variable can explain this behavior. The $z_3$ signal spends long intervals near zero, punctuated by sharp intermittent excursions. As a result, in the reconstructed space, a $k$-NN neighborhood may contain points that are close to each other but actually belong to different phases of the intermittent spikes. After $n$-step pushforward propagation, these phase differences can be amplified, producing a broad empirical future cloud $\widehat K^n_{N,k}(\mathbf x,\cdot)$. Thus, a large empirical closure score does not necessarily indicate a genuinely multi-valued future kernel; it may also arise from a finite neighborhood radius combined with finite-time dynamical stretching.

Moreover, we perform a sensitivity analysis of ranking stability with respect to the number of neighbors $k$, the pushforward $n$, the query-set size $N_{\mathcal{Q}}$, the size of the Theiler window $w$, and additive measurement noise. For each configuration, we compute the Spearman rank correlation between the resulting ranking and the baseline ranking presented in Table~\ref{tab:future_cloud_baseline_comparison}. We find that the relative ranking remains robust across variations in $k$, $n$, the Theiler window size $w$, and $N_{\mathcal Q}$, even though the absolute values of $\widehat{\mathcal E}_{n,k,N}^{*}$ change, as expected for a finite-resolution and finite-horizon quantity. Under additive noise, the ranking remains moderately to strongly correlated with the baseline ranking, while the absolute scale increases because noise broadens the empirical future cloud. Detailed results are reported in Appendix~\ref{ape:Robust}.        

\subsection{Real data illustration}

We apply the finite-resolution estimator to two open real-world time-series datasets: double pendulum trajectories~\cite{kaheman2019learning, hirsh2021structured} and measles outbreak data~\cite{london1973recurrent}. We use it as a finite-horizon, finite-resolution closure score for comparing different observable choices and transformations.     

\textbf{Double Pendulum:} We first consider simulated and measured trajectories of a double pendulum. The Lagrangian of this system is: 
\begin{equation}
    \mathcal{L} = \frac{1}{2}(m_1(\dot{x}_1^2 + \dot{y}_1^2) + m_2(\dot{x}_2^2 + \dot{y}_2^2)) + \frac{1}{2}(I_1\dot{\theta}_1^2 + I_2\dot{\theta}_2^2) - (m_1\mathbf{y}_1 + m_2\mathbf{y}_2)g,
\end{equation}
where $x_1 = a_1 \sin(\theta_1)$, $x_2 = l_1\sin(\theta_1) + a_2\sin(\theta_2)$, $\mathbf{y}_1 = a_1\cos(\theta_1)$ and $\mathbf{y}_2 = l_1\cos(\theta_1) + a_2\cos(\theta_2)$, $\theta_1$ and $\theta_2$ are the angles that the top and bottom pendulum arms make with the vertical axis, $m_1$ and $m_2$ are the masses distributed over each arm, $l_1$ and $l_2$ are the lengths, $a_1$ and $a_2$ are the distances from the joints to the centres of mass of the arms, and $I_1$ and $I_2$ are the moments of inertia for each arm. Using the Euler-Lagrange equations, we can construct two second-order differential equations for this chaotic system:
\begin{equation}
    \frac{\rm{d}}{\rm{d}t}\left(\frac{\partial\mathcal{L}}{\partial \dot{\theta}_i}\right) - \frac{\partial\mathcal{L}}{\partial \theta_i} = 0, \quad i = 1,2. 
\end{equation}
The data are sampled at $\Delta t = 0.001$s for both the simulated and the experimental datasets. 

Similarly, we can write down the expression for the differential embedding of this four-dimensional dynamical system. Instead of calculating and comparing these differential embeddings one by one, we compare the following types of differential embeddings built from the same measured angle: 
\begin{equation}\label{eq: Differential embeddings}
\begin{aligned}
    F_{\theta_i} &= (\theta_i, \dot{\theta}_i, \ddot{\theta}_i, \dddot{\theta_i}) \\
    F_{\sin{\theta}_i} &= (\sin\theta_i, \dot{\sin\theta}_i, \ddot{\sin\theta}_i, \dddot{\sin\theta}_i) \\
\end{aligned}
\end{equation}
The observable $\sin\theta_i$ is less well conditioned near $\cos\theta = 0$. Across both trajectory datasets, we observe $\widehat{\mathcal E}_{n,k,N}^{*}(F_{\theta_i}) < \widehat{\mathcal E}_{n,k,N}^{*}(F_{\sin\theta_i}), i=1,2,$ over the tested horizons, as expected. This example shows the potential to select a better measurement function using this estimator.   

\textbf{Measles outbreaks:} We use measles outbreak data from New York City between 1928 and 1964, which has been shown to exhibit chaotic behavior~\cite{schaffer1985strange, sugihara1990nonlinear} as our second example. This time series contains intermittent outbreak bursts, with long intervals near a low baseline and high-amplitude events, which is similar to observations in the $z_3$-direction of the R\"ossler system. We compare the time-delay reconstruction built from the raw counts $x_t$ and from the transformed series $\log(1+x_t)$. For each pair of delay parameters $(m, \tau)$, we compute $\widehat{\mathcal{E}}_{n,k,N}^{*}$ using the same estimator.  
 
We observe that the log transformation gives lower $\widehat{\mathcal E}_{n,k,N}^{*}$ over all combinations of delay parameters, suggesting that it reduces burst-dominated future-cloud dispersion in the finite-resolution setting. The results are shown in Figure~\ref{fig:Real World Datasets}.    
\begin{figure}[t]
    \centering
    \includegraphics[width=0.95\linewidth]{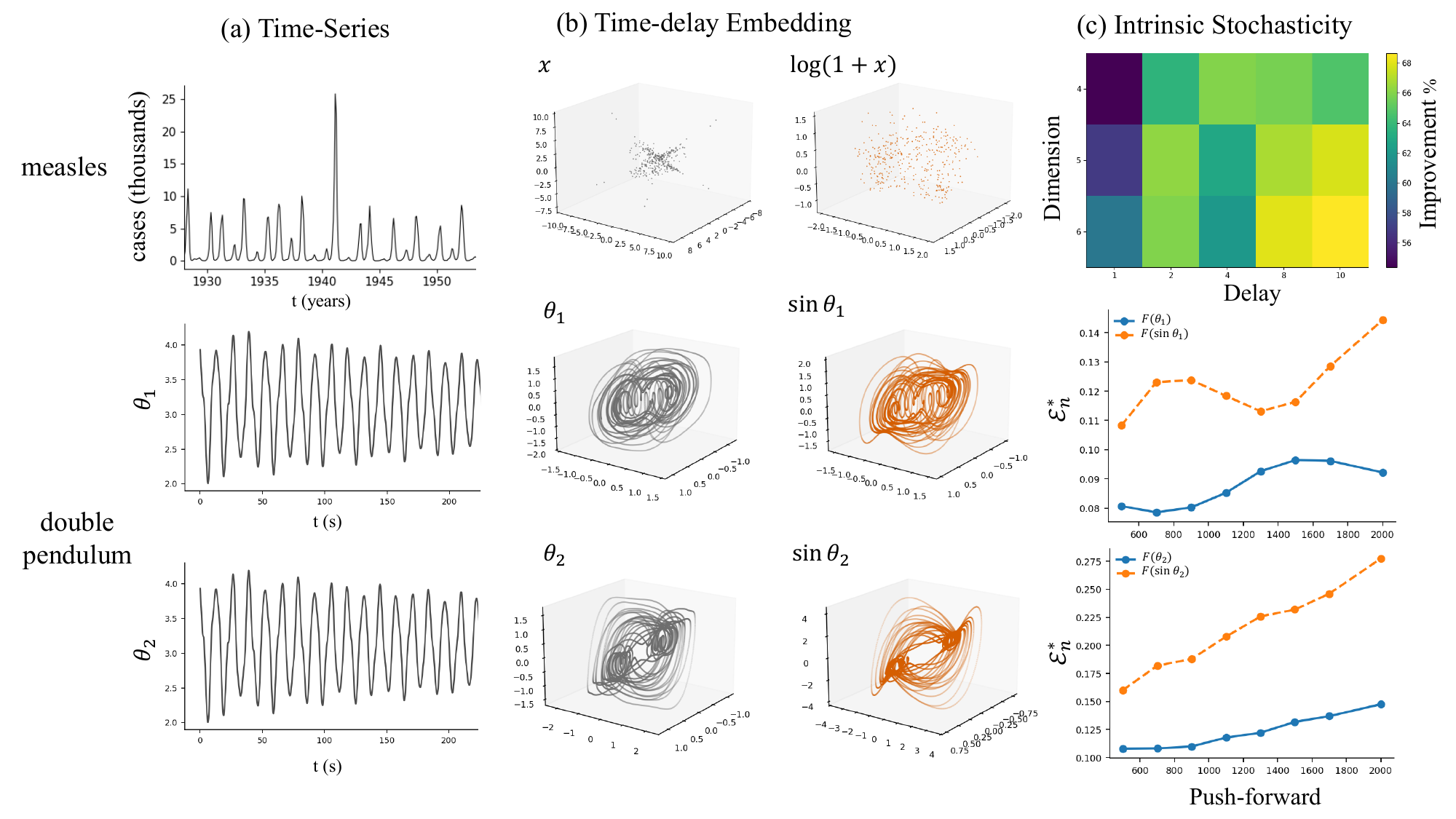}
    \caption{Illustrative real-data examples. (a). Raw time series. (b). Reconstructed point clouds. For the measles data, we compare time-delay reconstructions of the raw counts $x_t$ and the transformed series $\log(1+x_t)$, visualized by the first three principal components. For the double pendulum, we compare reconstructions built from $\theta_i$ and $\sin\theta_i$ using differential mappings. (c). $\widehat{\mathcal{E}}_{n,k,N}^{*}$ at finite resolution. For measles, we report the relative change: $\frac{\widehat{\mathcal E}_{n,k,N}^{*}(x) - \widehat{\mathcal E}_{n,k,N}^{*}(\log(1+x))}{\widehat{\mathcal E}_{n,k,N}^{*}(x)}$. For the double pendulum, we compare $\widehat{\mathcal{E}}^{*}_{n,k}({\theta_i})$ and $\widehat{\mathcal{E}}^{*}_{n,k}(\sin{\theta_i})$ across pushforward horizons. These examples illustrate finite-resolution differences in closure between observable choices.} 
\label{fig:Real World Datasets}
\end{figure}

\subsection{Downstream task impact}\label{Downstream task impact}
We next examine whether the finite-resolution closure score $\widehat{\mathcal E}_{n,k,N}^{*}$ aligns with downstream predictive performance and provides an explanation for the observed discrepancy in performance from a non-deterministic evolutionary perspective.  

The core intuition is that a deterministic model can only generate a single future prediction from each input. If the induced kernel $K^n(\mathbf x,\cdot)$ is non-Dirac, the same reconstructed state $\mathbf{x}$ is associated with multiple possible $n$-step futures. Consequently, no deterministic predictor can make the pointwise conditional risk vanish for such $\mathbf{x}$. More precisely, for any measurable deterministic predictor $g: \mathcal{X} \to \mathcal{X}$, the $n$-step risk is:
\begin{equation}
\mathcal R_n(g)
=
\int_{\mathcal X}
\int_{\mathcal X}
d_{\mathcal X}\bigl(g(\mathbf x),\mathbf y'\bigr)
K^n(\mathbf x,d\mathbf y')\,d\mu(\mathbf x).
\end{equation}
By the definition of $m_{n}$, we have $\mathcal R_n(g)\ge \mathcal E_{n}^*$, which means $\mathcal{E}_{n}^*$ gives a lower bound on the first-moment deterministic closure risk.

However, the empirical estimator $\widehat{\mathcal E}_{n,k,N}^{*}$ is only a finite-resolution and finite-sample proxy for this closure risk. We can use it as a pre-training diagnostic to compare candidate observables, rather than as a certification of downstream model performance. Rollout errors also depend on the model class, the dictionary or neural architecture, regularization, optimization, and the amount of training data. In the experiments below, we examine whether observables with smaller empirical closure scores tend to yield easier learning problems for two representative learning frameworks: Extended Dynamic Mode Decomposition (EDMD)~\cite{williams2015data} and Delay Invariant Measure (DIM)~\cite{botvinick2025invariant}.    

\textbf{Extended Dynamic Mode Decomposition:} We first compare EDMD models trained on time-delay reconstructions constructed from the $z_1$ and $z_3$ coordinate projections of the R\"ossler system. The EDMD model approximates the Koopman operator by lifting state variables into a dictionary of nonlinear observables and fitting a finite-dimensional linear operator~\cite{korda2018linear,klus2016data, lusch2018deep}. 

For each observable, we construct a time-delay reconstruction and use it as the state variable for EDMD. In both experiments, third-order polynomial features of the delay coordinates (e.g., $z_1, {z_1}^2, {z_1}^3, z_{1}^2z_2,...$) serve as the dictionary, and ridge regression is employed to fit the Koopman operator. The EDMD model is trained on one reconstructed trajectory, while a separate trajectory generated from a different initial condition is used for validation. Performance is evaluated by rollout NRMSE over multiple horizons. The results are presented in Figure~\ref{fig:EDMD}. 

\begin{figure}[!ht]
    \centering
    \includegraphics[width=1\linewidth]{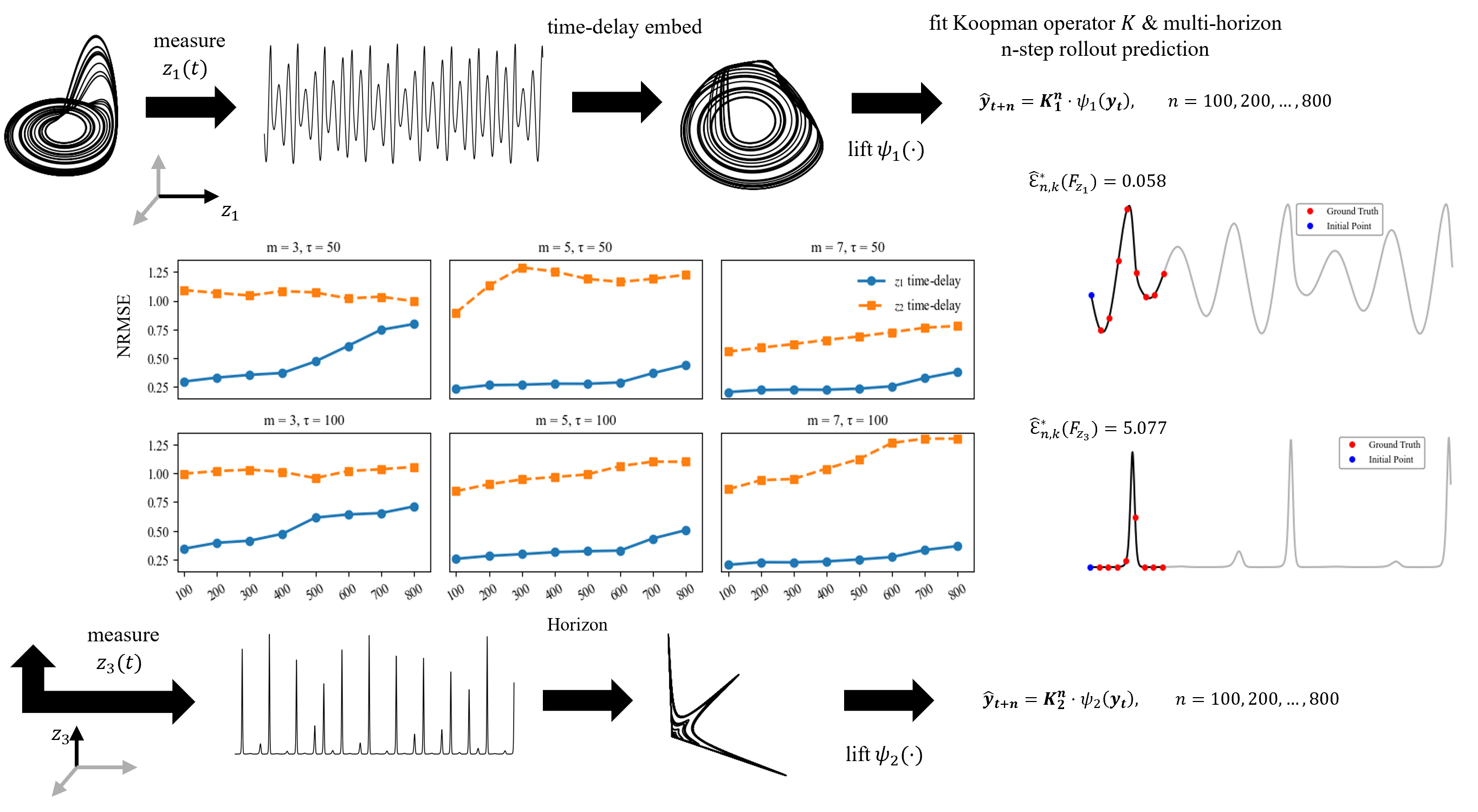}
    \caption{EDMD rollout prediction from two time-delay reconstructions of the R\"ossler system built from the $z_1$- and $z_3$-coordinate projections. EDMD is trained with third-order polynomial observables and ridge regression, and rollout performance is measured by NRMSE on a separate validation trajectory. The $z_1$-based reconstruction has a smaller empirical closure score $\widehat{\mathcal E}_{n,k,N}^{*}$ than the $z_3$-based reconstruction,  and it also yields lower EDMD rollout error in this experiment. This agreement is consistent with the interpretation of $\widehat{\mathcal E}_{n,k,N}^{*}$ as a finite-resolution closure diagnostic. }
\label{fig:EDMD}
\end{figure}

\textbf{Delay Invariant Measure Matching:} We also test the same observable comparison using the Delay Invariant Measure (DIM) objective~\cite{botvinick2025measure}. DIM is motivated by the result that, under the assumptions of the DIM framework, matching invariant measures in the same time-delay reconstructed space identifies the dynamics up to topological conjugacy. Since this learning objective depends on the geometry of the chosen reconstruction, a reconstruction with larger finite-resolution closure ambiguity may make the learning problem harder. 

Let $\hat{\nu}$ denote the empirical invariant measure of the full-state R\"ossler system, $\hat{\nu}_{\theta_j}$ denote the empirical invariant measure generated by the learned model $T_{\theta_j}$ and $T_{\theta_j}^{n}$ denote the $n$-step pushforward. Moreover, the invariant measure on the space reconstructed by $D_{z_j,\tau_j,m}$ is $\hat{\xi}^{m}_{z_j}$, and the loss function is:
\begin{equation}
\mathcal{L}_j(\theta_j) =\operatorname{ED}\!\left(T^n_{\theta_j}\#\hat{\nu}_{\theta_j},\hat{\nu}\right)
+
\operatorname{ED}\!\left(
D_{z_j,\tau_j,m}\#\hat{\nu}_{\theta_j},
\hat{\xi}^{m}_{z_j}
\right),
\qquad j\in\{1,3\},
\end{equation}
where $\operatorname{ED}$ denotes the energy distance between probability measures. We use the same architecture, optimizer, learning rate, number of epochs, and delay dimension. The results are presented in Figure~\ref{fig:DIM}. 
\begin{figure}[!ht]
\centering
\includegraphics[width=1\linewidth]{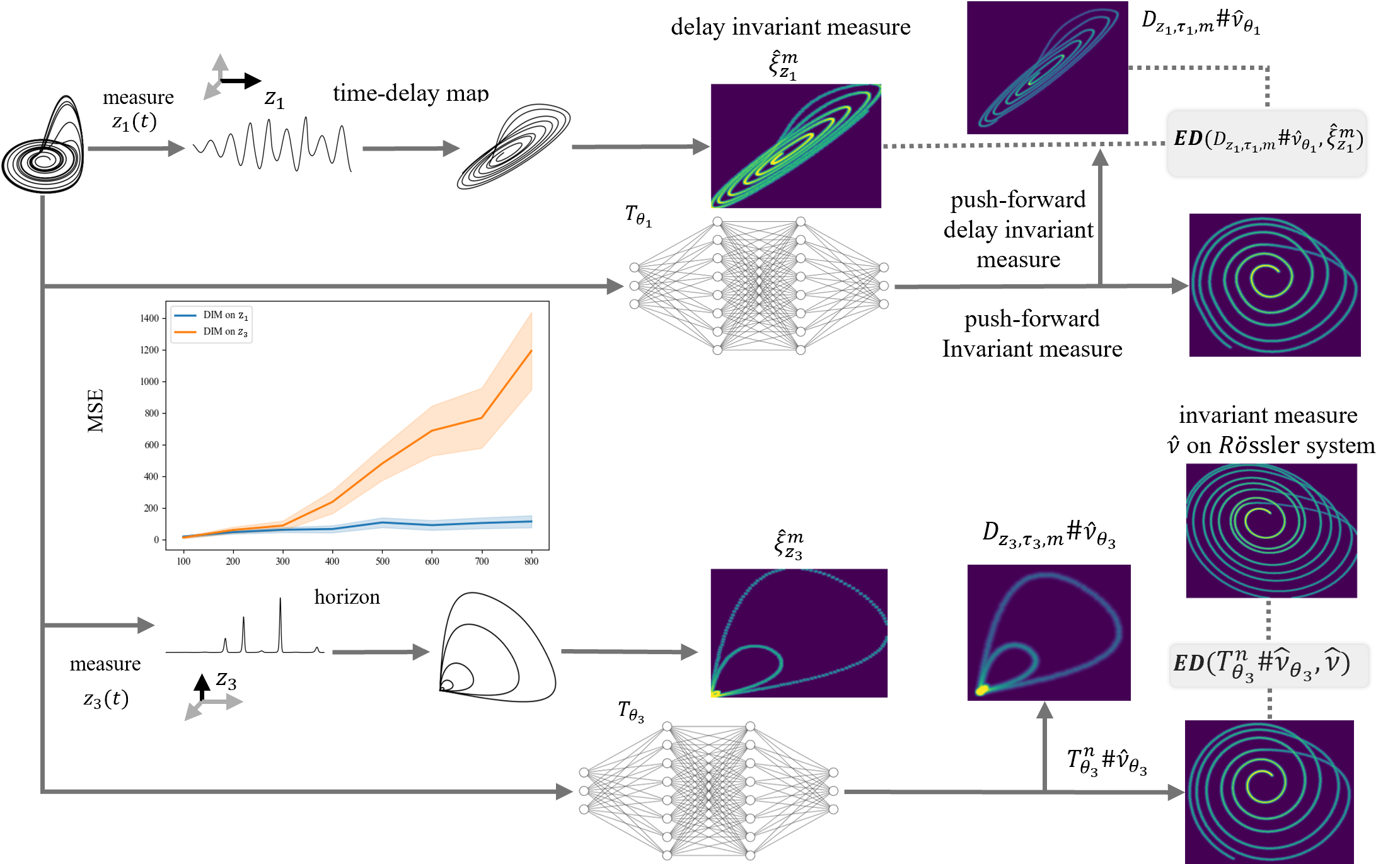}
\caption{Rollout prediction from $z_1$- and $z_3$-based time-delay reconstructions of the R\"ossler system. Both models use the same MLP architecture with 256 hidden units, Tanh activation, and Adam optimization. The $z_1$-based reconstruction has a smaller empirical closure score $\widehat{\mathcal E}_{n,k,N}^{*}$ than the $z_3$-based reconstruction, and it also yields lower rollout error over the tested horizons.} 
\label{fig:DIM}
\end{figure}

\section{Discussion and conclusion}

In this paper, we investigate the conditions under which the reconstruction map is non-injective, multiple latent states map to the same reconstructed state and may separate in the future, making the reconstructed system no longer single-valued. To quantify the ambiguity introduced by the multi-valued evolution, we proposed a measure-theoretic framework that leverages the invariant measure on the dynamical system and models the induced $n$-step evolution as an induced conditional kernel that describes the conditional distribution of futures given the reconstructed state $\mathbf{x}$. 

Inspired by this insight, we introduced the intrinsic stochasticity
$\mathcal{E}^*_{n}$ that serves as a measurement of the $n$-step deterministic closure, and reveals the geometric mechanism that governs the occurrence of information loss. For numerical implementation, we use the $k$-NN as the estimator to approximate $\mathcal{E}^*_{n}$ at a finite resolution under finite samples. Numerical examples support the use of $\widehat{\mathcal E}_{n,k,N}^{*}$ as a pre-training diagnostic for observable selection and provide a new perspective to understand the failure of the time-delay map for nonlinear dynamical system reconstruction. This viewpoint also provides insights for essential questions in related fields, including control and network science, e.g., how to select informative combinations of observables and how to determine the minimal number of sensors required to reconstruct high-dimensional dynamics without knowledge of the governing equations~\cite{liu2013observability, liu2016control, liu2011controllability, letellier2018symbolic}. Our future work focuses on the real-world applications, especially when the system is only partially observed. All code is available at https://github.com/CrispDyt/Wasserstein-Geometry-of-Information-Loss-in-Nonlinear-Dynamical-Systems.

\appendix 
\section{Differential embeddings of R\"ossler system}\label{ape: Rossler}

In this section, we provide details about the differential mappings of the R\"ossler system, which we used in Section~\ref{Rossler Case}. The formulation of R\"ossler system is:  
\begin{equation}
\begin{aligned}
\dot{z_1} &= -\,z_2 \;-\; z_3, \\
\dot{z_2} &= z_1 \;+\; a z_2, \\
\dot{z_3} &= b \;+\; z_3(z_1 - c) ,
\end{aligned}
\end{equation}
with the standard chaotic parameter values $a = 0.2, b = 0.2, c = 5.7$. Fix an observable $h$ and dimension as 3, the differential map is reconstructed as: 
\begin{equation}
D_{h,3}(\mathbf{z}) = (h(\mathbf{z}), \mathcal{L}_{f}h(\mathbf{z}), \mathcal{L}_{f}^2h(\mathbf{z}))
\end{equation}
For each case, we compute the Jacobian determinant and categorize each differential map as a diffeomorphism, as shown in Table~\ref{tab:rossler_differential_embeddings}. For compactness, define
\begin{equation}
\begin{aligned}
u_1&:=\dot z_1=-z_2-z_3,\\
u_2&:=\dot z_2=z_1+a z_2,\\
u_3&:=\dot z_3=b+z_3(z_1-c),
\end{aligned}
\qquad
\begin{aligned}
v_1&:=\ddot z_1=-z_1-a z_2-b-z_3(z_1-c),\\
v_2&:=\ddot z_2=a z_1+(a^2-1)z_2-z_3,\\
v_3&:=\ddot z_3=(b+z_3(z_1-c))(z_1-c)+z_3(-z_2-z_3).
\end{aligned}
\end{equation}

\begin{table}[!ht]
\centering
\footnotesize
\setlength{\tabcolsep}{11pt}
\renewcommand{\arraystretch}{1.15}
\caption{Explicit differential reconstructions of the Rössler system. 
The category indicates whether the corresponding reconstruction map is a global diffeomorphism in the ambient coordinates.}
\label{tab:rossler_differential_embeddings}
\begin{tabular}{lll}
\hline
\textbf{Observable} 
& \textbf{Differential mapping} 
& \textbf{Category} \\
\hline

\(h(\mathbf z)=z_1\)
& \((z_1,u_1,v_1)\)
& Non-diffeomorphism \\

\(h(\mathbf z)=z_2\)
& \((z_2,u_2,v_2)\)
& Diffeomorphism \\

\(h(\mathbf z)=z_3\)
& \((z_3,u_3,v_3)\)
& Non-diffeomorphism \\

\(h(\mathbf z)=z_1+z_2\)
& \((z_1+z_2,u_1+u_2,v_1+v_2)\)
& Non-diffeomorphism \\

\(h(\mathbf z)=z_1+z_3\)
& \((z_1+z_3,u_1+u_3,v_1+v_3)\)
& Non-diffeomorphism \\

\(h(\mathbf z)=z_2+z_3\)
& \((z_2+z_3,u_2+u_3,v_2+v_3)\)
& Non-diffeomorphism \\

\(h(\mathbf z)=(z_1,z_2)\)
& \((z_1,z_2,u_1)\)
& Diffeomorphism \\

\(h(\mathbf z)=(z_1,z_3)\)
& \((z_1,z_3,u_1)\)
& Diffeomorphism \\

\(h(\mathbf z)=(z_2,z_3)\)
& \((z_2,z_3,u_2)\)
& Diffeomorphism \\

\(h(\mathbf z)=(z_2,z_3)\)
& \((z_2,z_3,u_3)\)
& Non-diffeomorphism \\

\hline
\end{tabular}
\end{table}

\section{Sensitivity analysis}\label{ape:Robust}

To investigate the influence of parameters on the ranking, we performed sensitivity analysis over parameters, including the neighborhood size $k$, the pushforward horizon $n$, the Theiler window size $w$, the size of the query set $N_{\mathcal{Q}}$, and additive measurement noise, and we calculated the Spearman rank correlation between the default setting and the changed configuration in Table~\ref{tab:sensitivity_ranking}.       

\begin{table}[!ht]
\centering
\caption{Sensitivity of the reconstruction ranking induced by $\widehat{\mathcal E}_{n,k,N}^{*}$. Spearman correlations are computed against the baseline setting $n=20$, $k=50$, $N_{\mathcal Q}=5000$, and Theiler window $30$. We also report the median relative change rate ($\textbf{Median}$) and the maximum relative change rate $\textbf{Maximum}$ in the absolute $\widehat{\mathcal E}_{n,k,N}^{*}$ values across reconstructions.}
\label{tab:sensitivity_ranking}
\begin{tabular}{lccc}
\hline
\footnotesize
\textbf{Configuration}
&
\textbf{Spearman}
&
\textbf{Median}
&
\textbf{Max}
\\
\hline
\(k=20\) & \(1.00\) & \(0.42\) & \(0.52\) \\
\(k=100\) & \(1.00\) & \(0.55\) & \(0.62\) \\
\(n=5\) & \(0.99\) & \(0.41\) & \(0.70\) \\
\(n=10\) & \(1.00\) & \(0.15\) & \(0.43\) \\
\(n=40\) & \(1.00\) & \(0.20\) & \(0.49\) \\
\(w=0\) & \(1.00\) & \(0.03\) & \(0.07\) \\
\(w=10\) & \(1.00\) & \(0.02\) & \(0.06\) \\
\(w=60\) & \(1.00\) & \(0.01\) & \(0.04\) \\
\(N_{\mathcal Q}=1000\) & \(1.00\) & \(0.02\) & \(0.05\) \\
\(N_{\mathcal Q}=2500\) & \(1.00\) & \(0.02\) & \(0.03\) \\
\(N_{\mathcal Q}=10000\) & \(1.00\) & \(0.00\) & \(0.02\) \\
noise \(=0.01\) & \(0.85\) & \(0.10\) & \(0.78\) \\
noise \(=0.05\) & \(0.91\) & \(0.75\) & \(1.10\) \\
noise \(=0.10\) & \(0.98\) & \(0.91\) & \(2.59\) \\
\hline
\end{tabular}
\end{table}

\section{Implementation details for future-cloud baseline diagnostics}\label{app:future_cloud_baselines}
\setcounter{equation}{0}
\renewcommand{\theequation}{C \arabic{equation}}

Here, we give the implementation details for the three baseline future-cloud diagnostics used in Table~\ref{tab:future_cloud_baseline_comparison}. All three baselines are computed from the same Theiler-windowed $k$-nearest neighbor future clouds used in the computation of $\widehat{\mathcal E}_{n,k,N}^{*}$.  

Let $\{\tilde{\mathbf{x}}_t\}_{t=0}^{T-1}\subset\mathcal X$ be the reconstructed trajectory after coordinate standardization. For a query index $q$, let $\mathcal N_k(q)$ denote the set of indices of the $k$ nearest neighbors of $\tilde{\mathbf x}_q$, after applying the Theiler window. For a fixed pushforward horizon $n$, the corresponding standardized future cloud is: 
\begin{equation}
\tilde{\mathcal Y}_q = \{\tilde{\mathbf{y}}_{q,j}\}_{j=1}^{k} := \{\tilde{\mathbf x}_{t+n}:t\in\mathcal N_k(q)\}.
\end{equation}
We denote its centroid by $\bar{\tilde{\mathbf y}}_q = \frac{1}{k}\sum_{j=1}^{k}\tilde{\mathbf y}_{q,j}$. The empirical $n$-step future kernel at $\mathbf{x}_q$ is
\begin{equation}
\widehat K^n_{N,k}(\tilde{\mathbf x}_q,\cdot) =\frac{1}{k}\sum_{j=1}^{k}\delta_{\tilde{\mathbf{y}}_{q,j}},
\end{equation}
where each neighbor's future is assigned equal empirical mass $1/k$. Our empirical pointwise cost for the deterministic closure is:
\begin{equation}\label{definition_deterministic_score}
\widehat m_{n,k,N}(\tilde{\mathbf x}_q)=\inf_{\mathbf y\in\mathcal X}
\frac1k\sum_{j=1}^{k}
\|\mathbf y-\tilde{\mathbf y}_{q,j}\|_2.
\end{equation}

\noindent\textbf{Empirical conditional variance trace}. The empirical conditional variance trace is defined by
\begin{equation}
V_{n,k}(q)
= \frac{1}{k}
\sum_{j=1}^{k}
\|\mathbf y_{q,j}-\bar{\mathbf y}_q\|_2^2,
\end{equation}
which is a second-moment dispersion of the future cloud, whereas $\widehat m_{n,k,N}$ is a first-moment best-Dirac closure cost.

\noindent\textbf{Local short-horizon prediction error.} The local short-horizon prediction error uses $\bar{\mathbf y}_q$ as a local constant predictor for the actual future of the query point:
\begin{equation}
P_{n,k}(q) = \|\tilde{\mathbf x}_{q+n} -\bar{\tilde{\mathbf y}}_q\|_2.
\end{equation}
Compared with $\widehat m_{n,k,N}(\mathbf x_q)$, $P_{n,k}(q)$ depends on the future state $\mathbf x_{q+n}$ of the query point, thus a local prediction-error diagnostic rather than a function of the empirical future cloud alone. By averaging over all query points, it measures how well the neighbor-future centroid predicts the observed future trajectory. 

\paragraph{Mean pairwise future-cloud dispersion.}
The mean pairwise future-cloud dispersion is defined by
\begin{equation}
D^{\mathrm{pair}}_{n,k}(q) =\frac{2}{k(k-1)}\sum_{1\le i<j\le k}\|\mathbf y_{q,i}-\mathbf y_{q,j}\|_2,
\end{equation}
which measures the average separation between two distinct neighbor futures. Although $D^{\mathrm{pair}}_{n,k}(q)$ is a first-moment dispersion of the future cloud, it measures pairwise spread within the neighborhood of the future cloud while $\widehat m_{n,k, N}(\mathbf x_q)$ measures the distance from the future cloud to the closest Dirac representative.

All methods reported in Table~\ref{tab:future_cloud_baseline_comparison}
are obtained by averaging the pointwise values over the query set $\mathcal Q$ as:
\begin{equation}
\widehat{V}_{n,k,N}=\frac1{|\mathcal Q|}\sum_{q\in\mathcal Q}V_{n,k}(q),
\quad
\widehat P_{n,k,N} =\frac1{|\mathcal Q|}\sum_{q\in\mathcal Q}P_{n,k}(q),
\quad 
\widehat D^{\mathrm{pair}}_{n,k,N}
=
\frac1{|\mathcal Q|}
\sum_{q\in\mathcal Q}D^{\mathrm{pair}}_{n,k}(q).
\end{equation}

From Table~\ref{tab:future_cloud_baseline_comparison}, we observe that all four indices are closely related, have the same numerical trend, and this is not a coincidence since they are computed from the same empirical future kernel $\widehat K^n_{N,k}(\mathbf x_q,\cdot)$ but in different loss geometry and optimal minimizer. 

First, we discuss the relationship between $\widehat m_{n,k,N}(\mathbf x_q)$ and $V_{n,k}(q)$. $\widehat m_{n,k,N}(\mathbf x_q)$ is a first-moment best-Dirac loss, whereas $V_{n,k}(q)$ is a second-moment dispersion around the centroid $\bar{\mathbf y}_q$. In particular, conditional variance scales quadratically with separation, while $\widehat m_{n,k,N}$ scales linearly. Define $M_{n,k}(q)=\frac1k\sum_{j=1}^{k}\|\mathbf y_{q,j}-\bar{\mathbf y}_q\|_2$. Since the geometric median minimizes the
first-moment loss,
\begin{equation}
\widehat m_{n,k,N}(\mathbf x_q)
\le
M_{n,k}(q).
\end{equation}
By Jensen's inequality, $M_{n,k}(q)\le\sqrt{V_{n,k}(q)}$.
Consequently,
\begin{equation}\label{relation1}
\widehat m_{n,k,N}(\mathbf x_q)
\le
M_{n,k}(q)
\le
\sqrt{V_{n,k}(q)}.
\end{equation}
Equation~\eqref{relation1} explains why both metrics yield similar qualitative rankings, while $V_{n,k}(q)$ has a much larger scale when the future cloud has large separations.

Next, we discuss the relationship between $\widehat{m}_{n,k,N}(\mathbf{x}_q)$ and $D^{\mathrm{pair}}_{n,k}$. By the definition~\eqref{definition_deterministic_score}, for any fixed future state $\tilde{\mathbf y}_{q,i}$, we have
\begin{equation}
\widehat m_{n,k,N}(\tilde{\mathbf x}_q)
\le
\frac1k
\sum_{j=1}^{k}
\|\tilde{\mathbf y}_{q,i}-\tilde{\mathbf y}_{q,j}\|_2.
\end{equation}
Averaging over \(i=1,\ldots,k\) yields
\begin{equation}
\widehat m_{n,k,N}(\tilde{\mathbf x}_q)\le\frac1{k^2}\sum_{i=1}^{k}\sum_{j=1}^{k}\|\tilde{\mathbf y}_{q,i}-\tilde{\mathbf y}_{q,j}\|_2=\frac{k-1}{k}D^{\mathrm{pair}}_{n,k}(q).
\end{equation}

Conversely, let $\mathbf y_q^*$ be a geometric median of the future cloud. By the triangle inequality,
\begin{equation}
\|\tilde{\mathbf y}_{q,i}-\tilde{\mathbf y}_{q,j}\|_2
\le
\|\tilde{\mathbf y}_{q,i}-\mathbf y_q^*\|_2
+
\|\tilde{\mathbf y}_{q,j}-\mathbf y_q^*\|_2 .
\end{equation}
Averaging over \(i\) and \(j\) yields
\begin{equation}
\frac{1}{k^2}\sum_{i=1}^{k}\sum_{j=1}^{k}\|\tilde{\mathbf y}_{q,i}-\tilde{\mathbf y}_{q,j}\|_2\le2\widehat m_{n,k,N}(\tilde{\mathbf x}_q).
\end{equation}
Hence,
\begin{equation}
\widehat m_{n,k,N}(\tilde{\mathbf x}_q)
\le
\frac{k-1}{k}D^{\mathrm{pair}}_{n,k}(q)
\le
2\widehat m_{n,k,N}(\tilde{\mathbf x}_q),
\end{equation}
Thus, $D^{\mathrm{pair}}_{n,k}$ and $\widehat m_{n,k,N}(\tilde{\mathbf x}_q)$ are comparable up to constant factors, which explains why the two metrics have similar numerical scales with similar rankings across all reconstructions. 

The local short-horizon prediction error $P_{n,k}(q)$ serves as a prediction baseline, which takes the centroid $\bar{\tilde{\mathbf y}}_q$ as a local predictor and evaluates it on the actual $n$-step future of the query point, $\tilde{\mathbf x}_{q+n}$. Therefore, we use it to compare $\widehat m_{n,k,N}(\tilde{\mathbf x}_q)$ as a standard short-horizon prediction. 

\bibliographystyle{plain}
\bibliography{reference}

\end{document}